\documentclass[12pt, draftclsnofoot, onecolumn]{IEEEtran}

\normalsize
\ifCLASSINFOpdf
\else
\fi

\usepackage{amsfonts}
\usepackage{amsmath}
\usepackage{amssymb}
\usepackage{bbding}
\usepackage{graphicx}
\usepackage{subfigure}
\usepackage{algorithm, algorithmic}
\usepackage{booktabs}
\usepackage{multirow}
\usepackage{diagbox}
\usepackage{cite}
\usepackage{hyperref}
\usepackage{cleveref}

\newtheorem{lemma}{Lemma}
\newtheorem{proposition}{Proposition}
\newtheorem{proof}{Proof}
\DeclareMathOperator*{\argmax}{argmax}

\begin{document}

\title{Transmit Covariance and Waveform Optimization for Non-orthogonal CP-FBMA System}

\author{Yuhao~Qi,
        Jian~Dang,~\IEEEmembership{Member,~IEEE,}
        Zaichen~Zhang,~\IEEEmembership{Senior~Member,~IEEE,}
        Liang~Wu,~\IEEEmembership{Member,~IEEE,}
        and~Yongpeng~Wu,~\IEEEmembership{Senior~Member,~IEEE}
\thanks{This work was supported by the National Key R\&D Program of China under grant 2018YFB1801101 and 2016YFB0502202, NSFC projects (61971136, 61601119, 61960206005, and 61803211), Jiangsu NSF project (No. BK20191261), Zhejiang Lab (No. 2019LC0AB02), Young Elite Scientist Sponsorship Program by CAST (YESS20160042), and Zhishan Youth Scholar Program of SEU. Part of this paper was published in the 11th International Conference on Wireless Communications and Signal Processing, 2019.}
\thanks{Y. Qi, J. Dang, Z. Zhang, and L. Wu are with the National Mobile Communications Research Laboratory, Southeast University, Nanjing 210096, China. Z. Zhang is also with the Purple Mountain Laboratories, Nanjing 211111, China. Y. Wu is with the Department of Electronic Engineering, Shanghai Jiao Tong University, Minhang 200240, China (e-mail: qiyuhao@seu.edu.cn; dangjian@seu.edu.cn; zczhang@seu.edu.cn; wuliang@seu.edu.cn; yongpeng.wu@sjtu.edu.cn).}
\thanks{Corresponding authors: J. Dang (dangjian@seu.edu.cn) and Z. Zhang (zczhang@seu.edu.cn)}}

\maketitle

\begin{abstract}
Filter bank multiple access (FBMA) without subbands orthogonality has been proposed as a new candidate waveform to better meet the requirements of future wireless communication systems and scenarios. It has the ability to process directly the complex symbols without any fancy preprocessing. Along with the usage of cyclic prefix (CP) and wide-banded subband design, CP-FBMA can further improve the peak-to-average power ratio and bit error rate performance while reducing the length of filters. However, the potential gain of removing the orthogonality constraint on the subband filters in the system has not been fully exploited from the perspective of waveform design, which inspires us to optimize the subband filters for CP-FBMA system to maximizing the achievable rate. Besides, we propose a joint optimization algorithm to optimize both the waveform and the covariance matrices iteratively. Furthermore, the joint optimization algorithm can meet the requirements of filter design in practical applications in which the available spectrum consists of several isolated bandwidth parts. Both general framework and detailed derivation of the algorithms are presented. Simulation results show that the algorithms converge after only a few iterations and can improve the sum rate dramatically while reducing the transmission delay of information symbols.
\end{abstract}


\IEEEpeerreviewmaketitle

\section{Introduction}
Waveform design has become a hot topic in the fifth generation (5G) communications research to better satisfy the stringent requirements such as low out-of-band (OOB) radiation, robustness against carrier frequency offset (CFO) and flexible radio access \cite{5GNOW}. As \cite{ABMH19} pointed out, in order to meet the diverse requirements of 5G services, two primary candidate approaches had been proposed by researchers. The first approach, proposition of new waveforms for 5G, had not been adopted in current standards. Instead, the third generation partnership project (3GPP) adopted the second approach in Release-15 \cite{Release-15}, i.e., using different numerologies of the same parent waveform, cyclic prefix orthogonal frequency division multiplexing (CP-OFDM), mainly due to that OFDM has its own advantages and that OFDM was among the key ingredients in the commercialized fourth generation (4G) communication system. Besides, to meet the spectral emission requirement, windowing or filtering processing is required for plain CP-OFDM in 5G new radio (5G-NR). Compared with plain CP-OFDM, windowed or filtered OFDM (F-OFDM) \cite{JAJM15} can effectively reduce OOB emissions and be received with plain CP-OFDM receiver. Even though new waveforms were not selected by 3GPP for 5G-NR, these schemes remain as an interesting choice for future system development \cite{ALJM19}, as windowed or filtered OFDM has some limitations \cite{RNSM17} such as reduced spectral efficiency, and lower robustness in frequency selective channels. Furthermore, filtering or windowing may not provide as low OOB emissions as some new waveforms (e.g., FBMC). In fact, \cite{AYHA20} discussed this issue in the paradigm of sixth generation (6G) and addressed that multiple waveforms may be utilized together in the same frame for the next generation of wireless communications standards. Therein, examples and various reasons were also given to support this claim. Among various new waveforms including filter bank multi-carrier (FBMC) \cite{BFB11,FSTW14,TLLZ15}, universal-filtered multi-carrier (UFMC) \cite{VWSB13}, generalized frequency division multiplexing (GFDM) \cite{MMGC14}, etc., FBMC is a basic one and has been intensively studied. It consists of frequency-well-localized subband filters, and each data stream is processed individually by the filter in that subband. In fact, FBMC has been recognized as a competitive candidate in some new scenarios such as carrier aggregation, cognitive radio with spectrum sensing, and even in  dense wavelength multiplexing passive optical network (DWDM-PON) based fronthaul using a radio-over-fiber technique \cite{SJSJ19}. The most popular form of FBMC is offset quadrature amplitude modulation (FBMC/OQAM) \cite{BFB11,PCNL02}, which splits the input QAM symbols into real and imaginary parts and requires orthogonality between subbands in the real domain for improving the time-frequency localization of the waveform and to achieve maximal spectral efficiency. Other variants of FBMC include FBMC/QAM \cite{RZDL10,YCKZ15,MDMT15} and exponentially modulated filter bank (EMFB) \cite{JAMR03,AVIH09}. FBMC/QAM directly processes the input QAM symbols to avoid the two-dimensional intrinsic interference problem in FBMC/OQAM \cite{RZDL10}. EMFB shares similar principles with FBMC/OQAM but has more efficient implementation and simpler mathematical tractability \cite{AVIH09}. In this paper, we pick EMFB as a representative of conventional FBMC which maintains the orthogonality between subbands in the real domain. However, in practice, orthogonality is susceptible to non-ideal factors including channel distortion. In such a case, equalization must be employed with the aim of restoring the orthogonality. The equalization for conventional FBMC can be very simple, yet its performance is not satisfactory especially for channels with severe frequency selectivity and usually upper bounded by that of the equalized OFDM system. Besides, the real domain orthogonality means the power on the imaginary part of the received signal may not be well utilized for symbol detection. Moreover, orthogonal schemes are ineligible to achieve the entire capacity region of the channel. Taking the above issues into consideration, \cite{DZWW17} proposed a new framework of non-orthogonal FBMC system which drops the unnecessary orthogonality constraint imposed on the filter bank design, thus improving the equalization performance significantly which can be even much better than that of OFDM system. The input symbols to the synthesis filter bank (SFB) can now be complex ones with clear physical meaning and there is no need to split them into real and imaginary parts any more as in FBMC/OQAM. When used in multi-user communication, non-orthogonal FBMC is referred to as non-orthogonal filter bank multiple access (FBMA), where every subband is exclusively allocated to a specific user for access. \cite{DGZW17} introduced cyclic prefix (CP) and wide-banded subband design into non-orthogonal FBMA system. The so called CP-FBMA further brings several new desirable advantages such as improved equalization performance with much lower computational complexity, reduced length of filters and number of subcarriers, which lead to lower transmission latency and peak-to-average power ratio (PAPR), respectively. In addition, it has been proved that CP-FBMA is capacity achieving in certain system configurations \cite{DWZW18}.

A potential key advantage of utilizing non-orthogonal FBMA that has been recognized but not well exploited is the released freedom of filter design. Since non-orthogonality means that there is no stringent constraint on the filters any longer, a natural question is, whether we can optimize the waveform to further improve the system performance such as the achievable rate?

The answer is intuitively at least non-negative. Actually, when transferring the objective of filters design from satisfying real-domain orthogonality between subbands, which is usually a strict condition, to that of maximizing some direct performance metrics such as sum rate, it is highly possible that the investigated performance metric can be improved to some extent compared to the conventional orthogonal FBMA and non-orthogonal FBMA with non-optimized filters. To fully exploit the released design freedom, we need figure out what are the constraints that should impose to filters in this case. It turns out that there are not many constraints on filters except for length and energy. In fact, without orthogonality we may even extend the filters from wide-band in \cite{DGZW17} to full-band, as long as they improve the desired figure-of-merit of the system. Therefore, in this work, we firstly propose a filter optimization algorithm in time domain with length constraint under unit energy. Then, we extend the method and add stopband energy constraints to represent spectrum occupation requirement. As a preliminary work, we optimize the filters based on instantaneous channel state information (CSI) in this work with the aim of exposing the potential gain of non-orthogonal FBMA by filter optimization. In addition, with known CSI, the covariance matrices of users' symbol vectors should also be optimized. Thus, a joint optimization algorithm is also proposed in this work.

Since FBMC with non-orthogonal subband design is a relatively new concept, there is limited work on its waveform optimization. Most existing works are for prototype filter design. For instance, \cite{CQJH13} employed the $\alpha$-based branch-and-bound algorithm on the prototype filter optimization to minimize the stopband energy for FBMC systems. With the same objective, \cite{TCLJ19} considered the constraint on the time-domain channel estimation performance in terms of the mean squared error of estimated channel impulse response, and employed the method of Lagrange multipliers and the
Newton's method to optimize the prototype filter. \cite{HKHP17} suggested a simple waveform design for FBMC considering the time domain localization based on single prototype filter, which is more appropriate for the multi-path fading channel. \cite{TGXZ18} optimized the stopband of the prototype filter to reduce inter-carrier interference between adjacent subcarriers. \cite{RTTA19} proposed a prototype filter design based on convex optimization to minimize the OOB pulse energy. \cite{AFDB17} took the scenario of FBMC-based massive multiple-input multiple-output (MIMO) into account and proposed a prototype filter design method to improve the signal-to-interference-plus-noise ratio (SINR) performance of the system. Although the specific optimization methods in those works are of reference value, a common thing in those works is that the filters are in general designed to satisfy the perfect reconstruction (PR)/nearly perfect reconstruction (NPR) constraint \footnote{PR/NPR is interpreted as orthogonality from signal processing perspective rather than communication perspective.} or manage the self-interference level without PR/NPR constraint, and derived from a common prototype filter. On the other hand, in non-orthogonal FBMA, the subband filters may not have similar frequency domain structures and thus cannot be derived from a single prototype filter. Therefore, the filter design methods in existing works cannot be applied directly in this work.

Specifically, we allow each user in the system to occupy the full available bandwidth and propose a filter waveform design algorithm to maximize the sum rate according to the CSI. The time domain filter coefficients are selected as the optimization variables and obtained in an iterative procedure. In addition, we propose a manifold-based gradient ascent algorithm to solve the suboptimization problem encountered in the iterative process. On this foundation, we propose a joint waveform and covariance matrix optimization algorithm, where the filters and covariance matrices of the input symbol vectors are optimized alternatively in each iteration. Besides, the joint optimization algorithm is capable of accommodating for the stopband energy constraints, which is useful in practical applications with spectrum occupation restrictions.

For clarity, we summarize our contributions as follows:
\begin{itemize}
\item We propose a full-band waveform optimization algorithm for uplink CP-FBMA system to maximize the sum rate, which leads to a suboptimization problem. To cope with it, we propose a manifold-based gradient ascent algorithm to obtain a suboptimal solution.

\item We propose a joint waveform and covariance matrix optimization algorithm and take stopband energy constraints into account, which leads to two suboptimization problems. The first one is solved based on semidefinite program (SDP). The second one is solved based on generalized singular value decomposition (GSVD) and water-filling algorithm.
\end{itemize}

The rest of this paper is organized as follows. Section II presents the system model of non-orthogonal CP-FBMA. Section III proposes a waveform optimization algorithm on the full band. Section IV proposes a joint waveform and covariance matrix optimization algorithm. Section V provides simulation results. Section VI concludes this paper.

\emph{Notation}: Matrices are denoted by bold uppercase letters (e.g., $\mathbf{X}$). Vectors are represented by bold \emph{italic} letters (e.g., $\boldsymbol{X}$ or $\boldsymbol{x}$). Scalars are denoted by normal font (e.g., $x$). The real and complex number fields are denoted by $\mathbb{R}$ and $\mathbb{C}$, respectively. $\mathfrak{R}[\cdot]$ is the real part of a complex number. $\mathbb{E}[\cdot]$ represents the statistical expectation. $[\cdot]_i$ and $\Vert\cdot\Vert$ are the $i$th entry and the two-norm of a vector, respectively. $[\cdot]_{i,j}$ is the entry for the $i$th row and $j$th column of a matrix. $(\cdot)^T$ and $(\cdot)^H$ stand for transpose and Hermitian transpose. ${\rm tr}(\cdot)$, ${\rm rank}(\cdot)$ and $\vert\cdot\vert$ denote the trace, rank and determinant of a matrix. ${\rm diag}(\boldsymbol{x})$ represents the diagonal matrix whose main diagonal elements are given by the vector $\boldsymbol{x}$. $\mathbf{I}_i$ denotes the identity matrix of size $i\times i$, $\mathbf{W}_i$ denotes the $i$-point discrete Fourier transform (DFT) matrix, and $\mathbf{0}_{i\times j}$ denotes the all zero matrix of size $i\times j$.

\section{System Model}
In this section, we briefly review the system model of non-orthogonal CP-FBMA for further study. Fig. \ref{system model} shows the transceiver structure of an $M$-subband CP-FBMA system for $M$ users in the uplink, where each user occupies a wide subband rather than a group of narrow subbands as in conventional FBMA. For example, if the total number of users is $8$, the total number of subbands for conventional narrow-banded FBMA may be $128$ (each user occupies $16$ subbands), whereas the total number of subbands for wide-banded CP-FBMA is just $8$, which greatly reduces the number of subbands. In addition, in general, the length of the filter on each subband is a few times of the number of subbands (e.g., the filter length can be $4M$). For narrow-banded FBMA, the filter length could be $512$, whereas that of wide-banded CP-FBMA can be $32$, which also reduces the length of convolution per subband thus the complexity of the system. Although a single subband is employed in the wide-banded scheme, the occupied bandwidth of each user is still the same as that of the narrow-banded scheme.

\begin{figure*}
\setlength{\abovecaptionskip}{-0.5cm}
\centering
\includegraphics[width=6.2in]{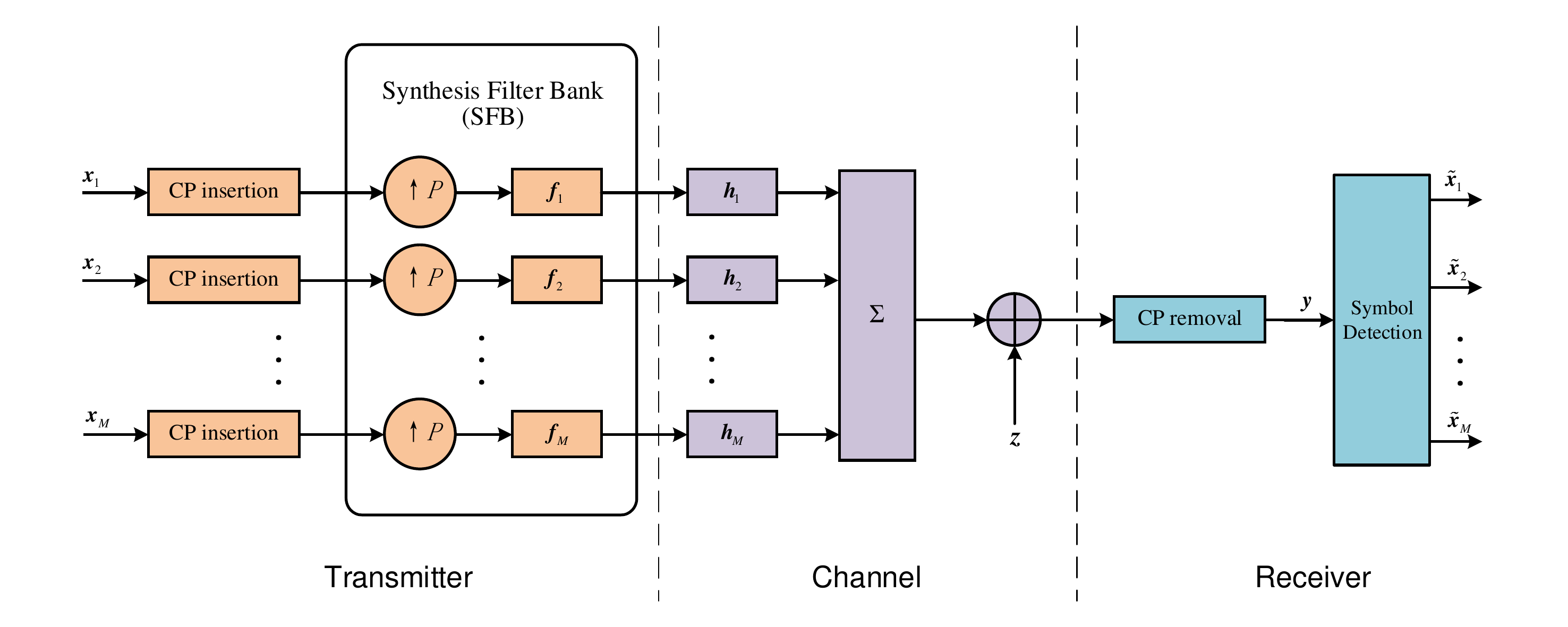}
\centering
\caption{Transceiver structure of CP-FBMA system with a total of $M$ subbands.}
\label{system model}
\end{figure*}

\subsection{Transmitter Structure}
At the transmitter side, the input QAM-modulated symbol block $\boldsymbol{x}_m$ on subband-$m$ is assumed to be of length $N$. It is appended by a CP of length $L_g$ and fed to the SFB for filtering with upsampling factor $P\leq M$. When $P=M$, the system is called critically sampled. When $P<M$, the system is called oversampled. The length of the synthesis filter $\boldsymbol{f}_m$ is $N_f$. Generally, $N_f$ is a few times of $M$. The signal on subband-$m$ is transmitted through a multi-path fading channel $\boldsymbol{h}_m$ of length $L_h$. In this paper, the channels $\boldsymbol{h}_1,...,\boldsymbol{h}_M$, are assumed to be perfectly known. Regarding channel estimation, some existing methods can be used, such as pilot-based channel estimation. It should be noted that CP is inserted on the symbols before SFB rather than after SFB as in \cite{DXTX15}. Therefore, the power spectral density of the transmit signals is solely determined by response of the filters and not be affected by CP. Besides, the effective length of CP after upsampling is $L_gP$ which should cover both the length of the physical channel and the length of the filter at the transmitter, i.e., $L_gP\geq N_f+L_h-1$.

The main differences between conventional FBMA and CP-FBMA are as follows: 1) in conventional FBMA, the source bits are directly mapped to real PAM symbols (for FBMC/OQAM, the modulated QAM symbols need to be split into real and imaginary parts), while in CP-FBMA, the transmit symbols are standard complex QAM symbols without any fancy preprocessing; 2) in conventional FBMA, the available bandwidth is divided into many narrow banded subbands, and a group of successive subbands is exclusively allocated to a user for access, while in CP-FBMA, the narrow subbands are replaced with a wide subband for each user; 3) in conventional FBMA, the upsampling factor equals to half the number of subbands, i.e., $P=M/2$, while in CP-FBMA, $P$ is a parameter that can be set more flexibly.

\subsection{Receiver Structure}
At the receiver side, the aggregated signal is corrupted by an additive white Gaussian noise (AWGN) whose samples $\boldsymbol{z}$ are i.i.d. complex Gaussian random variables with zero mean and variance $N_0$. The received signal after removing the effective CP is denoted by $\boldsymbol{y}$, which can be written as
\begin{align}\label{bmy}
\boldsymbol{y}=\sum_{m=1}^M \mathbf{H}_m \mathbf{F}_m \mathbf{U} \boldsymbol{x}_m+\boldsymbol{z},
\end{align}
where $\mathbf{H}_m$, $\mathbf{F}_m\in \mathbb{C}^{NP\times NP}$ are circular matrices whose first columns are given by zero-padded version of $\boldsymbol{h}_m$ and $\boldsymbol{f}_m$, respectively. $\mathbf{U}\in \mathbb{R}^{NP\times N}$ is the upsampling operator whose $(nP)$th rows ($n=0,...,N-1$) form an identity matrix of size $N\times N$ while all the other rows are zeros. Then the symbol detection module recovers the transmitted symbols from $\boldsymbol{y}$. Regarding the symbol detection scheme for non-orthogonal CP-FBMA system, the most recent work is proposed in \cite{DQZW19}, which uses iterative LMMSE estimation. To be more specific, $\mathbf{H}_m$ and $\mathbf{F}_m$ can be diagonalized, i.e., $\mathbf{H}_m=\mathbf{W}_{NP}^H\mathbf{\Lambda}_{\boldsymbol{H}_m}\mathbf{W}_{NP}$, $\mathbf{F}_m=\mathbf{W}_{NP}^H\mathbf{\Lambda}_{\boldsymbol{F}_m}\mathbf{W}_{NP}$, where $\mathbf{\Lambda}_{\boldsymbol{H}_m}$, $\mathbf{\Lambda}_{\boldsymbol{F}_m}\in\mathbb{C}^{NP\times NP}$ are diagonal matrices whose main diagonal entries are given by the $(NP)$-point DFT of $\boldsymbol{h}_m$ and $\boldsymbol{f}_m$, respectively. Thus, (\ref{bmy}) is rewritten as
\begin{align}
\boldsymbol{Y}=\sum_{m=1}^{M} \mathbf{\Lambda}_{\boldsymbol{H}_m} \mathbf{\Lambda}_{\boldsymbol{F}_m} \mathbf{W}_{NP} \mathbf{U} \mathbf{W}_N^H \mathbf{W}_N \boldsymbol{x}_m+\boldsymbol{Z}=\mathbf{T} \mathbf{W}_{(M)} \boldsymbol{x}+\boldsymbol{Z},
\end{align}
where $\boldsymbol{Y}=\mathbf{W}_{NP}\boldsymbol{y}$, $\boldsymbol{Z}=\mathbf{W}_{NP}\boldsymbol{z}$, $\mathbf{T}=[\mathbf{T}_1,...,\mathbf{T}_{M}]$ with $\mathbf{T}_m=\mathbf{\Lambda}_{\boldsymbol{H}_m} \mathbf{\Lambda}_{\boldsymbol{F}_m} \mathbf{W}_{NP} \mathbf{U} \mathbf{W}_N^H$, $\boldsymbol{x}=[\boldsymbol{x}_1^T,...,\boldsymbol{x}_{M}^T]^T$, and $\mathbf{W}_{(M)}$ is a block diagonal matrix whose $M$ subblocks on the main diagonal are all given by $\mathbf{W}_N$. Now, $\boldsymbol{x}$ can be estimated based on LMMSE estimator as
\begin{align}\label{hatx}
\boldsymbol{\hat{x}}=\mathbf{W}_{(M)}^H (\mathbf{T}^H \mathbf{T}+N_0 \mathbf{I}_{NP})^{-1} \mathbf{T}^H \boldsymbol{Y}.
\end{align}
With initial decision of $\boldsymbol{x}_m$, iterative multi-user interference cancellation can be performed to further improve the detection performance. It totally avoids the usage of analysis filter bank (AFB) at CP-FBMA receiver for the following reasons: 1) if AFB is kept, we will decode data on $a(\boldsymbol{y})$, where function $a(\cdot)$ denotes the operation of AFB; 2) if AFB is removed, we will decode data directly on $\boldsymbol{y}$, which is exactly the detection algorithm mentioned above. According to data processing inequality, the mutual information between the original information symbol $\boldsymbol{x}$ and $\boldsymbol{y}$, $a(\boldsymbol{y})$, satisfies $I(\boldsymbol{x};\boldsymbol{y})\geq I(\boldsymbol{x};a(\boldsymbol{y}))$, i.e., any operation on $\boldsymbol{y}$ will results in potential loss of information, unless $a(\cdot)$ is some sufficient statistics. However, there is no proof that AFB has such property. Though AFB is eliminated, some inherent filtering is still contained. On one hand, at the transmitter, the channel state information (CSI) of all users are considered in the filter design for each user, where some ``rough filtering'' is implicitly contained in this process. In other words, if one user would perceive strong interferers in a certain frequency band from the signals of other users, some inherent filtering is naturally considered in its filter optimization process to maximize the sum rate. On the other hand, the task of the receiver is to recover $\boldsymbol{x}$ from the immediate received signal $\boldsymbol{y}$. For any specific user, $\boldsymbol{x}$ can be regarded as the combination of the user's ``signal'' and its perceived ``interferers''. Therefore, any detection algorithm that can successfully recover $\boldsymbol{x}$ from $\boldsymbol{y}$ must have conducted some inherent filtering in its processing. Therefore, we can view the interference filtering as done implicitly in the proposed system. Moreover, avoiding using AFB has some additional advantages, including the important one that making the noise vector $\boldsymbol{z}$ in (\ref{bmy}) not be processed by filters so that it keeps an identity covariance matrix, which simplifies the waveform design in the following analysis.

\subsection{Processing Complexity}
In order to clarify the processing complexity of CP-FBMA, in this subsection, we show a numerical comparison with CP-OFDMA. Assume that the number of users is $M$ for both systems. In CP-OFDMA, at the transmitter, inverse fast Fourier transform (IFFT) is applied for the modulation of the OFDM signal for each user. Therefore, the processing complexity is $\mathcal{O}(M\text{log}_2M')$ per QAM symbol, where $M'$ is the number of subcarriers in CP-OFDMA; at the receiver, since FFT can be applied for the demodulation of the received OFDM signal, the processing complexity is $\mathcal{O}(\text{log}_2M')$ per QAM symbol, where the complexity of equalization $\mathcal{O}(1)$ is omitted.

In CP-FBMA, at the transmitter, the processing complexity comes from the linear convolution of the symbol blocks and synthesis filters, and the processing complexity is $\mathcal{O}(((N+L_g)N_f)/N)$ per QAM symbol; at the receiver, the computational cost is dominated the inversion of matrix $\mathbf{T}^H\mathbf{T}+N_0\mathbf{I}_{NP}$ and matrix products, whose complexity is $\mathcal{O}(N^3P^3)$ and $\mathcal{O}(NP^3+N^2P)$, respectively. Note that $\mathbf{T}^H\mathbf{T}+N_0\mathbf{I}_{NP}$ is a block matrix with $P\times P$ subblocks, and each subblock is a diagonal matrix of size $N$. By rearranging the entries, $\mathbf{T}^H\mathbf{T}+N_0\mathbf{I}_{NP}$ can be transformed to a block diagonal matrix. Specifically, define an interleaving matrix $\mathbf{\Omega}$ of size $NP\times NP$. For $\forall\omega\in \{0,...,NP-1\}$, if $\omega={\omega}_1P+{\omega}_2$ where ${\omega}_1$ and ${\omega}_2$ are non-negative integers, the $\omega$th row of $\mathbf{\Omega}$ is all zeros except for the $({\omega}_1+{\omega}_2N)$th entry. The trick is that $\mathbf{\Omega}(\mathbf{T}^H\mathbf{T}+N_0\mathbf{I}_{NP})\mathbf{\Omega}^T$ is a block diagonal matrix with $N$ subblocks on the main diagonal, and the size of each subblock is $P\times P$. Noting that $\mathbf{\Omega}^T\mathbf{\Omega}=\mathbf{I}_{NP}$, we can turn to computing the inversion of $\mathbf{\Omega}(\mathbf{T}^H\mathbf{T}+N_0\mathbf{I}_{NP})\mathbf{\Omega}^T$, which involves $N$ matrix inversions with size $P\times P$, rather than the inversion of $\mathbf{T}^H\mathbf{T}+N_0\mathbf{I}_{NP}$ directly. Therefore, the complexity of the matrix inversion reduces to $\mathcal{O}(NP^3)$, and the processing complexity at the receiver is $\mathcal{O}((P^3+NP)/M)$ per QAM symbol.

To investigate the processing complexity more intuitively, we present the evaluation of execution time of the signal processing under the parameter setting in Table \ref{tab1}.

\begin{table}[ht]
\caption{System parameters setup for execution time evaluation}
\renewcommand{\arraystretch}{1}
\begin{center}
{\begin{tabular}{|c|c|c|}\hline\label{tab1}
\textbf{Systems} & \textbf{Parameters} & \textbf{Values} \\\hline
\multirow{3}*{All}&Number of users $M$ & 8 \\
\cline{2-3}
&Bandwidth & 30.72MHz\\
\cline{2-3}
&Modulation format & 16-QAM\\
\hline

\multirow{3}*{CP-OFDMA}&Number of subcarriers $M'$ & 2048 \\
\cline{2-3}
&Frequency-spacing & 15kHz\\
\cline{2-3}
&Length of CP & 144\\
\hline

\multirow{3}*{CP-FBMA}&Number of subbands & 8 \\
\cline{2-3}
&Length of synthesis filters $N_f$ & 32 \\
\cline{2-3}
&Length of symbol block $N$ & 64 \\
\hline

\end{tabular}}{}
\end{center}
\end{table}

For a fair comparison, the total numbers of QAM symbols transmitted by the two systems are the same. The simulation is based on MATLAB 2018b in a regular 2.30 GHz laptop, and the channel model is $10$-
tap complex random Gaussian channel with equal power gains among the taps and the total power is normalized to unity. Table \ref{tab2} plots the above complexity analysis and experimental results of the average execution time at both transmitters and receivers.

\begin{table}[ht]
\small
\caption{Complexity and experimental results}
\renewcommand{\arraystretch}{1}
\begin{center}
{\begin{tabular}{|c|c|c|c|c|}\hline\label{tab2}
\multirow{2}*{\diagbox[width=12.5em]{Indexes}{Systems}} & \multicolumn{2}{|c|}{\textbf{CP-OFDMA}} & \multicolumn{2}{|c|}{\textbf{CP-FBMA}} \\
\cline{2-5}
& Transmitter & Receiver & Transmitter & Receiver \\
\hline

Complexity in $\mathcal{O}(\cdot)$ notation  & $M\text{log}_2M'$ & $\text{log}_2M'$ & $((N+L_g)N_f)/N$ & $(P^3+NP)/M$ \\
\hline

Average execution time & 0.655ms & 0.350ms & 1.921ms & 111.921ms \\
\hline

\end{tabular}}
\end{center}
\end{table}

It can be seen that the execution time at the transmitter of CP-FBMA is nearly $3$ times higher than that of CP-OFDMA. Besides, the time taken by CP-FBMA to recover the signal is much longer than that by CP-OFDMA due to the matrix inversions in the process of equalization. This is the price that CP-FBMA pays to improve the system performance.

\subsection{Performance Analysis}
In this subsection, we discuss the latency, spectral efficiency and coherence time of the channel as they are correlated with each other. Here, latency can refer to two cases. One is the transmission delay of a certain information symbol, which is related to the length of filters. We will show that by optimizing the filters, the sum rate of CP-FBMA can be improved while reducing the filter length, thus reducing the transmission delay per symbol. The other latency refers to the decoding delay of a transmitted data block. In this regard, the proposed CP-FBMA requires the entire reception of the block of $N$ QAM symbols before decoding. Thus, it introduces a delay quantified by $(N+L_g)PT_s$, where $T_s$ is the sampling interval and other sources of delays such as channel delay or processing delay are ignored. Nevertheless, by using the multi-tap equalizer and maximum-likelihood sequence estimation, the proposed system also has the ability to start to decode as soon as the first symbols of the block are received (just as the conventional FBMA using per subband single tap equalization), but the receiving scheme presented in this paper does not support this.


As mentioned above, the required time of transmitting $MN$ QAM symbols is $(N+L_g)PT_s$. Hence, the spectral efficiency of CP-FBMA is $\frac{N\log_2K}{N+L_g}$ (bits/s/Hz), where $K$ is the QAM modulation order and $P$ is set as $M$. Under the same parameters, the spectral efficiency of conventional FBMA is $\log_2K$ (bits/s/Hz) due to no CP is inserted, which is higher than that of CP-FBMA. However, the length of CP depends on $N_f$ and $L_h$. When they are determined, the minimum length of CP can also be determined. If $N$ is relatively large, the portion of CP diminishes, and spectral efficiency of the two systems are approximately equal.

Lastly, we simply analyze the influence of coherent time of the channel on signal transmission. We will use some relatively specific parameters to explain what the size of the data block $N$ will be if its duration is within the coherence time, or how many data blocks can be transmitted within the coherence time if $N$ is given. Take extended pedestrian A (EPA) channel model in long term evolution (LTE) as an example. Assume the bandwidth $B=30.72$MHz, the coherence time of the channel $\tau_c=10$ms and the number of users $M=8$. The length of channel can be approximately calculated by $L_h={\tau}_{max}/T_s\approx 12$, where $\tau_{max}=410$ns is the maximum excess tap delay of EPA, and $T_s=1/B\approx 32.6$ns is the sampling interval. Thus, the length of CP can be evaluated by $L_g=\lceil (N_f+L_h-1)/P\rceil=6$, where $\lceil x\rceil$ denotes the minimum integer no smaller than $x$ and we assume the length of synthesis filters $N_f=4M$ and upsampling factor $P=M$. If the duration of the block of $N$ data symbols $(N+L_g)PT_s$ is within the coherence time $\tau_c$, it should satisfy $(N+L_g)PT_s\leq\tau_c$. Then we obtain $N\leq38394$, which means that approximately $38394$ symbols can be transmitted in the coherence time. On the other hand, assuming $N$ is fixed as $140$, the duration of the block is $\tau_d=(N+L_g)PT_s\approx 0.038$ms, and there are $\tau_c/\tau_d\approx 263$ blocks within the coherence time. Therefore, a sufficient amount of QAM symbols can be transmitted in the coherence time.

\section{Waveform Optimization On The Full Bandwidth}
In previous works, the synthesis filters of CP-FBMA system are band-pass filters in general. Each user occupies a subband, and there is little overlap with other subbands except for adjacent ones. The wide-banded design was originally proposed to reduce the filter length and complexity of symbol detection while improving PAPR and bit error ratio (BER) performance, as compared to the narrow-banded filter bank design. Related simulation results can be found in \cite{DZWW17,DGZW17,DQZW19}. To further improve the system performance, we extend the wide-banded design to full-band design in this paper, i.e., each filter is allowed to cover the full available bandwidth. In this case, overlap between subbands is much more severe, and the conventional filters without any optimization may suffer serious performance loss. Therefore, it is urgent to call for a new filter design method. In this section, we consider the waveform optimization on the full bandwidth to maximize the sum rate of CP-FBMA system. Firstly, we formulate the optimization problem and solve it by an outer iteration. In each iteration, the original problem leads to a suboptimization problem. Then, we propose a manifold-based gradient ascent algorithm to solve the suboptimization problem with inner iterations.

\subsection{Problem Formulation}
As indicated by \cite{DZWW17}, we will compute the achievable rate by observing $\boldsymbol{y}$ to avoid potential information loss. The signal model of (\ref{bmy}) is similar to that of a multi-user MIMO in uplink \cite[section 10.2]{DTPV05}. Therefore, the mutual information (with unit bits) by observing $\boldsymbol{y}$ could be computed similarly to that in \cite[section 10.2]{DTPV05}. Note that in our system this amount of information bits is transmitted in a total time of $(N+L_g)PT_s$ seconds and bandwidth of $1/T_s$ Hertz. Thus, the achievable sum rate of CP-FBMA (with unit bits/s/Hz) can be characterized by
\begin{align}\label{rate}
\sum_{m=1}^MR_m=\frac{1}{(N+L_g)P}\log_2\left|\mathbf{I}_{NP}+\frac{1}{N_0}\sum_{m=1}^M \mathbf{H}_m \mathbf{F}_m \mathbf{U} \mathbf{C}_m \mathbf{U}^T \mathbf{F}_m^H \mathbf{H}_m^H \right|,
\end{align}
where $\mathbf{C}_m\in \mathbb{C}^{N\times N}$ denotes the covariance matrix of the input symbol vector $\boldsymbol{x}_m$, which is defined as $\mathbf{C}_m\triangleq\mathbb{E}[\boldsymbol{x}_m\boldsymbol{x}_m^H]$. In a true MIMO system, the user wise average power constraint $P_m$ is imposed directly on $\boldsymbol{x}_m$ by $\frac{1}{N}\text{tr}(\mathbf{C}_m)=P_m$. Different from that, in CP-FBMA, the power constraint should be imposed on the SFB precoded symbol $\mathbf{F}_m\mathbf{U}\boldsymbol{x}_m$ through
\begin{align}\label{Cmconstraint}
\frac{1}{NP}{\rm tr}(\mathbb{E}[(\mathbf{F}_m\mathbf{U}\boldsymbol{x}_m)(\mathbf{F}_m\mathbf{U}\boldsymbol{x}_m)^H])
=\frac{1}{NP}{\rm tr}(\mathbf{F}_m \mathbf{U} \mathbf{C}_m \mathbf{U}^T \mathbf{F}_m^H)=P_m, m=1,...,M,
\end{align}
as $\mathbf{F}_m\mathbf{U}\boldsymbol{x}_m$ rather than $\boldsymbol{x}_m$ is the finally transmitted signal consuming that power. Considering that the optimal waveform occupies the full bandwidth, there is only energy constraint on the filters, i.e.,
\begin{align}
\boldsymbol{f}_m^H\boldsymbol{f}_m=1, m=1,...,M.
\end{align}
In general, the choice of $\mathbf{C}_m$ has impact on the optimal waveform $\boldsymbol{f}_m$. Therefore, they should be jointly optimized. We leave this joint optimization problem to Section \uppercase\expandafter{\romannumeral4}. In this section, we consider a simpler case, where $\mathbf{C}_m, m=1,...,M$, are identity matrices. We have the following lemma:
\begin{lemma}
For $\forall m\in\{1,...,M\}$, if $\mathbf{C}_m$ and $\boldsymbol{f}_m$ satisfy $\mathbf{C}_m=PP_m\mathbf{I}_N$ and $\boldsymbol{f}_m^H\boldsymbol{f}_m=1$, respectively, they also satisfy $\frac{1}{NP}{\rm tr}(\mathbf{F}_m \mathbf{U} \mathbf{C}_m \mathbf{U}^T \mathbf{F}_m^H)=P_m$.
\end{lemma}
\begin{proof}
See Appendix A.
\end{proof}
For the convenience of expression, we define
\begin{align}
R_{sum}\triangleq \log_2\left|\mathbf{I}_{NP}+\sum_{m=1}^M \frac{1}{N_0}\mathbf{H}_m \mathbf{F}_m \mathbf{U} \mathbf{C}_m \mathbf{U}^T \mathbf{F}_m^H \mathbf{H}_m^H \right|
\end{align}
Thus, the waveform design problem can be formulated as
\begin{subequations}
\begin{align}\label{optproblem2}
({\rm P1}):\mathop{\max}_{\boldsymbol{f}_m,\forall m}&R_{sum}\\
\textit{s.t.}\quad &\boldsymbol{f}_m^H\boldsymbol{f}_m=1, m=1,...,M,
\end{align}
\end{subequations}
where $\mathbf{C}_m=PP_m\mathbf{I}_N, m=1,...,M$. The coefficient $\frac{1}{(N+L_g)P}$ in the objective function is omitted in analysis as it will not affect the solutions $\boldsymbol{f}_1^*,...,\boldsymbol{f}_M^*$, but it is considered during simulations as the actual value of $R_{sum}$ would be affected.

\subsection{Problem Simplification}
By diagonalizing circular matrices $\mathbf{H}_m$ and $\mathbf{F}_m$, the objective function of $({\rm P1})$ can be rewritten as
\begin{align}\label{optproblem3}
R_{sum}=\log_2\left|\mathbf{I}_{NP}+ \sum_{m=1}^M \mathbf{\Lambda}_{\boldsymbol{H}_m} \mathbf{\Lambda}_{\boldsymbol{F}_m} \mathbf{Q}_m \mathbf{\Lambda}_{\boldsymbol{F}_m}^H \mathbf{\Lambda}_{\boldsymbol{H}_m}^H \right|,
\end{align}
where $\mathbf{Q}_m\triangleq\frac{PP_m}{N_0}\mathbf{W}_{NP}\mathbf{U}\mathbf{U}^T\mathbf{W}_{NP}^H$ turns out to be a block matrix consisting of $P\times P$ identical subblocks, and each subblock is expressed as $\frac{PP_m}{N_0}\mathbf{I}_N$. $\mathbf{Q}_m$ can be transformed into the following form:
\begin{align}\label{Qm}
\mathbf{Q}_m={ \sum_{n=0}^{N-1} \boldsymbol{q}_{m,n} \boldsymbol{q}_{m,n}^T },
\end{align}
where $\boldsymbol{q}_{m,n}\in\mathbb{R}^{NP}$ is a sparse vector, with a single non-zero element $[\boldsymbol{q}_{m,n}]_{iN+n}=\sqrt{\frac{PP_m}{N_0}}$ for $i=0,...,P-1$. Therefore, (\ref{optproblem3}) can be further simplified as
\begin{align}\label{DFTF}
R_{sum}=\log_2\left|\mathbf{I}_{NP}+ \sum_{m=1}^M\sum_{n=0}^{N-1} \mathbf{\Lambda}_{\boldsymbol{H}_m} \mathbf{\Lambda}_{\boldsymbol{q}_{m,n}} \boldsymbol{F}_m \boldsymbol{F}_m^H \mathbf{\Lambda}_{\boldsymbol{q}_{m,n}} \mathbf{\Lambda}_{\boldsymbol{H}_m}^H \right|
\end{align}
where $\mathbf{\Lambda}_{\boldsymbol{q}_{m,n}}\in\mathbb{C}^{NP\times NP}$ is a diagonal matrix whose main diagonal is given by $\boldsymbol{q}_{m,n}$, $\boldsymbol{F}_m\in \mathbb{C}^{NP}$ is the $(NP)$-point DFT of $\boldsymbol{f}_m$. Note that $\boldsymbol{F}_m$ can be written as $\boldsymbol{F}_m=\sqrt{NP}\mathbf{W}_{NP}\mathbf{P}\boldsymbol{f}_m$, where $\mathbf{P}=\begin{bmatrix}\mathbf{I}_{N_f}&\mathbf{0}_{N_f\times (NP-N_f)}\end{bmatrix}^T$. Thus, (\ref{DFTF}) can be reformulated as
\begin{align}\label{optproblem4}
R_{sum}=\log_2\left|\mathbf{I}_{NP}+ \sum_{m=1}^M \sum_{n=0}^{N-1} \mathbf{G}_{m,n} \boldsymbol{f}_m \boldsymbol{f}_m^H \mathbf{G}_{m,n}^H \right|,
\end{align}
where $\mathbf{G}_{m,n}\triangleq\sqrt{NP}\mathbf{\Lambda}_{\boldsymbol{H}_m} \mathbf{\Lambda}_{\boldsymbol{q}_{m,n}}\mathbf{W}_{NP}\mathbf{P}\in\mathbb{C}^{NP\times N_f}$ is a sparse matrix with only $(iN+n)$th rows are nonzero, $i=0,...,P-1$.

It is difficult to maximize the sum rate in (\ref{optproblem4}) by optimizing the filters on all subbands simultaneously. With a similar structure, \cite{YRh04} showed that the sum-rate maximization problem could be solved efficiently by using an iterative algorithm. In each iteration, only the rate of a single user is maximized by treating other users' interference as known variables. When the sum rate converges, iteration stops. In other words, optimal $\boldsymbol{f}_m$ can be found iteratively by regarding other filters as fixed, and in each iteration it can be written as
\begin{align}\label{fm2}
\boldsymbol{f}_m^{*}=\mathop{\argmax}_{\boldsymbol{f}_m^H\boldsymbol{f}_m=1}
{\log_2\left| \sum_{n=0}^{N-1} \mathbf{G}_{m,n} \boldsymbol{f}_m \boldsymbol{f}_m^H \mathbf{G}_{m,n}^H +\mathbf{\Phi}_m \right|},
\end{align}
where $\mathbf{\Phi}_m=\mathbf{I}_{NP}+\sum_{i=1,i\neq m}^M \sum_{n=0}^{N-1} \mathbf{G}_{i,n} \boldsymbol{f}_i \boldsymbol{f}_i^H \mathbf{G}_{i,n}^H$ represents the multi-user interference contribution and is known and fixed. It can be easily proved that $\mathbf{\Phi}_m$ is a block matrix with $P\times P$ subblocks which are diagonal matrices of size $N\times N$. For further derivation, we have the following lemma:
\begin{lemma}
$\mathbf{\Phi}_m^{-1/2}$ has the same block matrix structure with $\mathbf{\Phi}_m$.
\end{lemma}
\begin{proof}
See Appendix B.
\end{proof}

Thus, (\ref{fm2}) can be further reduced as
\begin{align}\label{fm3}
\boldsymbol{f}_m^{*}=\mathop{\argmax}_{\boldsymbol{f}_m^H\boldsymbol{f}_m=1}
{\log_2\left| \mathbf{I}_{NP}+ \sum_{n=0}^{N-1} \mathbf{\bar{G}}_{m,n} \boldsymbol{f}_m \boldsymbol{f}_m^H \mathbf{\bar{G}}_{m,n}^H \right|},
\end{align}
where $\mathbf{\bar{G}}_{m,n}\triangleq\mathbf{\Phi}_m^{-1/2}\mathbf{G}_{m,n}$ has the same matrix structure with $\mathbf{G}_{m,n}$ due to Lemma 2, i.e., only $(iN+n)$th rows are nonzero, $i=0,...,P-1$. This inspires us to remove the zero rows of $\mathbf{\bar{G}}_{m,n}$. By using the aforementioned interleaving matrix $\mathbf{\Omega}$, $\sum_{n=0}^{N-1} \mathbf{\bar{G}}_{m,n} \boldsymbol{f}_m \boldsymbol{f}_m^H \mathbf{\bar{G}}_{m,n}^H$ can be transformed to a block diagonal matrix, i.e., $\mathbf{\Omega}(\sum_{n=0}^{N-1} \mathbf{\bar{G}}_{m,n} \boldsymbol{f}_m \boldsymbol{f}_m^H \mathbf{\bar{G}}_{m,n}^H)\mathbf{\Omega}^T$ is a block diagonal matrix whose $n$th main diagonal subblock is $\mathbf{\tilde{G}}_{m,n} \boldsymbol{f}_m \boldsymbol{f}_m^H \mathbf{\tilde{G}}_{m,n}^H$, where $\mathbf{\tilde{G}}_{m,n}\in \mathbb{C}^{P\times N_f}$ evolves from $\mathbf{\bar{G}}_{m,n}$ by removing its zero rows. Thus, (\ref{fm3}) can be reformulated as
\begin{align}\label{fm4}
\boldsymbol{f}_m^{*}&=\mathop{\argmax}_{\boldsymbol{f}_m^H\boldsymbol{f}_m=1}
{\sum_{n=0}^{N-1}\log_2\left| \mathbf{I}_P+ \mathbf{\tilde{G}}_{m,n} \boldsymbol{f}_m \boldsymbol{f}_m^H \mathbf{\tilde{G}}_{m,n}^H \right|}
=\mathop{\argmax}_{\boldsymbol{f}_m^H\boldsymbol{f}_m=1}
{\sum_{n=0}^{N-1}\log_2\left(1+ \boldsymbol{f}_m^H \mathbf{B}_{m,n} \boldsymbol{f}_m \right)},
\end{align}
where $\mathbf{B}_{m,n}\triangleq\mathbf{\tilde{G}}_{m,n}^H\mathbf{\tilde{G}}_{m,n}\in\mathbb{C}^{N_f\times N_f}$ is a Hermitian matrix. Now, the original problem $({\rm P1})$ is transferred to the following suboptimization problem in each outer iteration:
\begin{subequations}
\begin{align}\label{optproblem2-sub}
({\rm P1-1}):\mathop{\max}_{\boldsymbol{f}_m}\quad&\bar{R}_m(\boldsymbol{f}_m)= \sum_{n=0}^{N-1}\log_2(1+\boldsymbol{f}_m^H\mathbf{B}_{m,n}\boldsymbol{f}_m)\\
\textit{s.t.}\quad &\boldsymbol{f}_m^H\boldsymbol{f}_m=1.
\end{align}
\end{subequations}

\subsection{Manifold-based Algorithm}
Problem $({\rm P1-1})$ is not convex as it maximizes a convex objective function with an equality constraint of waveform energy. Here, we use a manifold-based gradient ascent algorithm to find a near-optimal $\boldsymbol{f}_m$ \cite{QDZW19}. The manifold-based algorithm uses an iterative method to get the current solution. Different from the traditional gradient ascent algorithm, the solution $\boldsymbol{f}_m^{(t+1)}$ after the $(t+1)$th iteration in the manifold-based algorithm is written as
\begin{align}\label{update}
\boldsymbol{f}_m^{(t+1)}=\mathcal{R}_{\boldsymbol{f}_m^{(t)}}(\rho^{(t)}{\rm grad}_{\bar{R}_m}^{(t)})
=\frac{\boldsymbol{f}_m^{(t)}+\rho^{(t)}{\rm grad}_{\bar{R}_m}^{(t)}}
{\Vert \boldsymbol{f}_m^{(t)}+\rho^{(t)}{\rm grad}_{\bar{R}_m}^{(t)}\Vert}
\end{align}
where $\rho$ is the stepsize, ${\rm grad}_{\bar{R}_m}\in\mathbb{C}^{N_f}$ is the gradient of $\bar{R}_m$ on the manifold $\mathcal{M}=\{\boldsymbol{f}_m\vert\boldsymbol{f}_m^H\boldsymbol{f}_m=1\}$, which can be regarded as the projection of ${\nabla}_{\bar{R}_m}$ on the tangent space $T_{\boldsymbol{f}_m}(\mathcal{M})=\{\boldsymbol{\zeta}\vert{\boldsymbol{\zeta}}^H\boldsymbol{f}_m+\boldsymbol{f}_m^H\boldsymbol{\zeta}=0\}$ at $\boldsymbol{f}_m$, where ${\nabla}_{\bar{R}_m}\in\mathbb{C}^{N_f}$ is the gradient of $\bar{R}_m$ and can be calculated as
\begin{align}\label{nabla}
{\nabla}_{\bar{R}_m}={\frac{1}{\ln2}}{\sum_{n=0}^{N-1} \frac{\mathbf{B}_{m,n}\boldsymbol{f}_m}{1+\boldsymbol{f}_m^H\mathbf{B}_{m,n}\boldsymbol{f}_m}}.
\end{align}
\begin{proposition}
The closed form of ${\rm grad}_{\bar{R}_m}$ is ${\rm grad}_{\bar{R}_m}={\nabla}_{\bar{R}_m}-\mathfrak{R}[{({\nabla}_{\bar{R}_m})}^H\boldsymbol{f}_m]\boldsymbol{f}_m$.
\end{proposition}
\begin{proof}
See Appendix C.
\end{proof}

At the $(t+1)$th iteration, the stepsize $\rho^{(t)}$ in (\ref{update}) can be chosen approximately by backtracking line search \cite{SBLV04}. When ${\Vert{\rm grad}_{\bar{R}_m}\Vert}^2$ is smaller than $\epsilon$ which is a small constant, iteration stops. Besides, $\boldsymbol{f}_m$ after the final iteration is taken as the solution to Problem $({\rm P1-1})$.

For clarity, we summarize the above waveform optimization algorithm on the full bandwidth for solving $({\rm P1})$ in Algorithm \ref{waveopt1}.
\begin{algorithm}
	\caption{Waveform Optimization Algorithm on The Full Bandwidth for Solving $({\rm P1})$}
    \renewcommand{\algorithmicrequire}{\textbf{Input:}}
    \renewcommand{\algorithmicensure}{\textbf{Output:}}
	\label{waveopt1}
	\begin{algorithmic}[1]
        \REQUIRE $M$, $N$, $P$, $N_0$, $P_m$, $\boldsymbol{h}_m, m=1,...,M$
        \ENSURE Optimal $\boldsymbol{f}_m, m=1,...,M$
        \STATE Initialize $\boldsymbol{f}_m$ and construct $\mathbf{U}$, $\mathbf{H}_m$ and $\mathbf{F}_m, m=1,...,M$.
        \REPEAT
        \FOR {$m=1$ to $M$}
		\STATE Compute $\mathbf{B}_{m,n}, n=0,...,N-1$ in (\ref{fm4}), and set a small constant $\epsilon$.
		\REPEAT
        \STATE Compute ${\rm grad}_{\bar{R}_m}^{(t)}$ at $\boldsymbol{f}_m^{(t)}$ by ${\rm grad}_{\bar{R}_m}^{(t)}={\nabla}_{\bar{R}_m}^{(t)}-\mathfrak{R}[{({\nabla}_{\bar{R}_m}^{(t)})}^H\boldsymbol{f}_m^{(t)}]\boldsymbol{f}_m^{(t)}$.
        \STATE Choose $\rho^{(t)}$ by backtracking line search.
        \STATE Set $\boldsymbol{f}_m^{(t+1)}=\mathcal{R}_{\boldsymbol{f}_m^{(t)}}(\rho^{(t)}{\rm grad}_{\bar{R}_m}^{(t)})$.
        \UNTIL {$\Vert{\rm grad}_{\bar{R}_m}\Vert^2\leq\epsilon$}
        \ENDFOR
        \UNTIL {the sum rate converges}
	\end{algorithmic}
\setlength{\belowdisplayskip}{3pt}
\end{algorithm}

\subsection{Complexity Analysis}
In this subsection we analyze the complexity of Algorithm \ref{waveopt1}. As assumed in Section II, there are $MN$ QAM symbols to be transmitted in total. In a single outer iteration, the main computational cost before the inner iteration (from Step 5 to Step 9) comes from Step 4, which involves the computations of matrix $\mathbf{G}_{m,n}$, $\mathbf{\Phi}_m^{-1}$, $\mathbf{\bar{G}}_{m,n}$ and $\mathbf{B}_{m,n}$. The complexity of these computations are $\mathcal{O}(MNPN_f)$, $\mathcal{O}(MNPN_f+MNP^3)$, $\mathcal{O}(MNP^2N_f)$ and $\mathcal{O}(MNPN_f^2)$, respectively. Fortunately, by using the Woodbury identity \cite{GHCV96}, i.e., $(\mathbf{A}+\boldsymbol{b}\boldsymbol{b}^H)^{-1}=\mathbf{A}^{-1}-\mathbf{A}^{-1} \boldsymbol{b}(1+\boldsymbol{b}^H\mathbf{A}^{-1}\boldsymbol{b})^{-1} \boldsymbol{b}^H\mathbf{A}^{-1}$, the complexity of computing $\mathbf{\Phi}_m^{-1}$ can be reduced. For example, $\mathbf{\Phi}_2=\mathbf{\Phi}_1+\sum_{n=0}^{N-1}\mathbf{G}_{1,n}\boldsymbol{f}_1\boldsymbol{f}_1^H \mathbf{G}_{1,n}^H-\sum_{n=0}^{N-1}\mathbf{G}_{2,n}\boldsymbol{f}_2\boldsymbol{f}_2^H \mathbf{G}_{2,n}^H$, where $\mathbf{\Phi}_1^{-1}$ is known and $\boldsymbol{f}_1$ has been optimized but $\boldsymbol{f}_2$ has not. When computing $\mathbf{\Phi}_2^{-1}$, we can firstly compute  $(\mathbf{\Phi}_1+\mathbf{G}_{1,0}\boldsymbol{f}_1\boldsymbol{f}_1^H \mathbf{G}_{1,0}^H)^{-1}$ by $(\mathbf{\Phi}_1+\mathbf{G}_{1,0}\boldsymbol{f}_1\boldsymbol{f}_1^H\mathbf{G}_{1,0}^H)^{-1}= \mathbf{\Phi}_1^{-1}-\mathbf{\Phi}_1^{-1} \mathbf{G}_{1,0}\boldsymbol{f}_1(1+\boldsymbol{f}_1^H\mathbf{G}_{1,0}^H\mathbf{\Phi}_1^{-1} \mathbf{G}_{1,0}\boldsymbol{f}_1)^{-1} \boldsymbol{f}_1^H\mathbf{G}_{1,0}^H\mathbf{\Phi}_1^{-1}$, then compute $(\mathbf{\Phi}_1+\sum_{n=0}^1\mathbf{G}_{1,n}\boldsymbol{f}_1\boldsymbol{f}_1^H \mathbf{G}_{1,n}^H)^{-1}$ according to $(\mathbf{\Phi}_1+\mathbf{G}_{1,0}\boldsymbol{f}_1\boldsymbol{f}_1^H \mathbf{G}_{1,0}^H)^{-1}$ until $\mathbf{\Phi}_2^{-1}$ is obtained. Thus, the complexity of computing $\mathbf{\Phi}_m^{-1}$ is reduced to $\mathcal{O}(MNPN_f+MNP^2)$. Totally, the complexity of Step 4 is dominantly given by $\mathcal{O}(MNP^2N_f+MNPN_f^2)$ which is relatively deterministic.

The strict convergence analysis of the manifold-based algorithm may still remain an open problem. Thus, we turn to analyzing the computational cost within a single inner iteration where the dominated cost comes from Step 7. It involves matrix products, and the complexity is $\mathcal{O}(T_{se}MNN_f^2)$, where $T_{se}$ is the needed number of iterations of the backtracking line search. Therefore, the normalized complexity of the proposed algorithm per QAM symbol is $\mathcal{O}(T_{out}(P^2N_f+PN_f^2+T_{in}T_{se}N_f^2))$, where $T_{out}$ and $T_{in}$ are the total needed numbers of outer and inner iterations, respectively. The exact values of $T_{in}$ and $T_{se}$ are largely determined by $\epsilon$ and the initial value of $\rho^{(t)}$ (we denote the initial value as $\rho_0^{(t)}$), respectively. We observe in our simulations that when $\epsilon=1$ and $\rho_0^{(t)}=0.01$, the average value of $T_{in}$ is about $7$ and $T_{se}$ is about $1$. In this case, about $92.63\%$ execution time of Algorithm 1 is spent on Step 4 and about $0.85\%$ on Step 7. When $\epsilon=10^{-5}$ and $\rho_0^{(t)}=0.01$, the average value of $T_{in}$ is about $152$ and $T_{se}$ is still about $1$. In this case, about $73.59\%$ execution time is spent on Step 4 and about $12.09\%$ on Step 7. Although the portion of the execution time of Step 7 increases when $\epsilon=10^{-5}$, the improved sum rate is only $0.23\%$ higher than that when $\epsilon=1$. Therefore, $\epsilon$ can be set relatively large in practical applications, which means the number of inner iterations $T_{in}$ can be relatively small and thus leads to a lower complexity. Besides, it can be seen that the normalized complexity is independent of $N$ and largely depends on $P$ and $N_f$, where $P\leq M$. Due to the usage of wide-banded subbands for multi-user access, the number of subbands $M$ is relatively small. Moreover, by using the optimal waveform, even shorter filters can present high performance, which will be shown in simulations.

Regarding the signaling exchange process, the proposed optimization scheme can be implemented in such a procedure: users send pilot signals to the base station (BS), and BS estimates the channel impulse response according to the pilot signals and designs filters. Then the time domain coefficients of the designed filters are fed back to each user through the downlink control channel. As the computation is carried out in BS, the complexity would be more affordable, and more simple waveform optimization algorithm can be further studied in future works,

\begin{figure}
\setlength{\abovecaptionskip}{-0.3cm}
\centering
\includegraphics[width=3in]{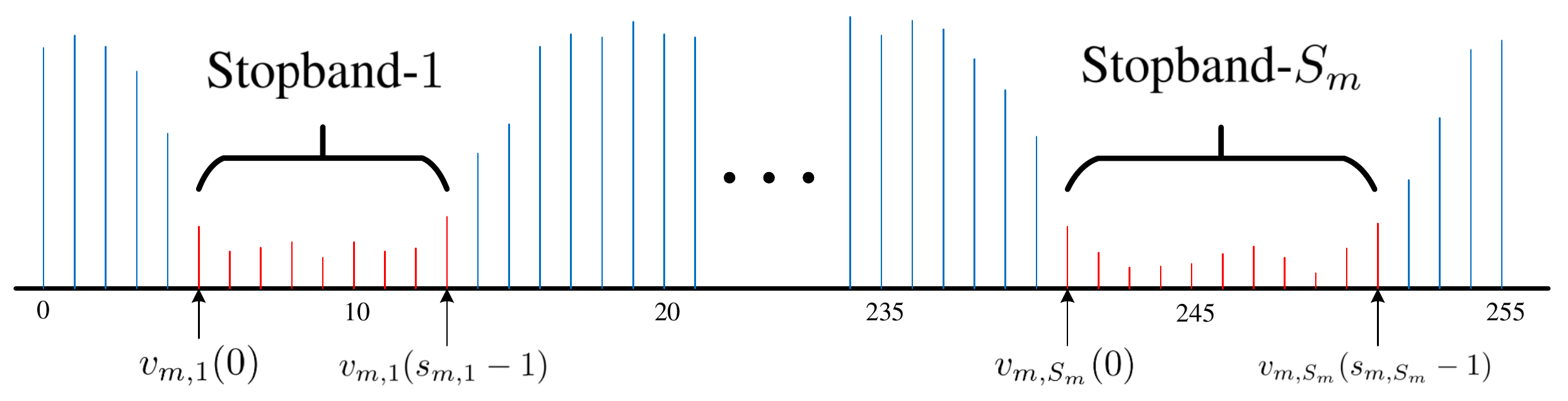}
\caption{The schematic diagram of $256$-point DFT of $\boldsymbol{f}_m$.}
\label{FFT}
\end{figure}

\section{Joint Waveform and Covariance Matrix Optimization}
In the above section, the covariance matrices $\mathbf{C}_m, m=1,...,M$, are chosen as fixed identity matrices. Although this assumption simplifies the analysis, it is not optimal in the sense that, the optimal covariance matrices have mutual dependence with the choice of the filters and are in general not identity ones. Besides, there will be constraints on the spectrum resources occupied by filters in practical applications. For instance, in cognitive radio, the transceiver is required to be capable of utilizing all available bandwidth resources which may consist of several isolated band fragments. Requirements on the spectrum occupation is transformed into per user stopband energy constraints in this paper. To fully reveal the performance gain of non-orthogonal CP-FBMA design, in this section we investigate the joint waveform and covariance matrices optimization to maximize the sum rate, and take the constraints on the stopband energy of filters into consideration.

\subsection{Stopband Energy Constraint}
In this subsection, we transform the spectrum occupation requirements into stopband energy constraints. For user-$m$, suppose there are $S_m$ stopbands and the number of frequency bins on the $i$th stopband is $s_{m,i}$. $\boldsymbol{v}_{m,i}=[v_{m,i}(0),...,v_{m,i}(s_{m,i}-1)]$ is a row vector where $v_{m,i}(n), n=0,...,s_{m,i}-1$ denote the positions of these frequency bins in $(NP)$-point DFT of $\boldsymbol{f}_m$. For clarity, we take the schematic diagram of $256$-point DFT of $\boldsymbol{f}_m$ as an example, which is shown in Figure \ref{FFT}. Define matrix $\mathbf{V}_{m,i}\in \mathbb{R}^{s_{m,i}\times NP}$ whose elements are all zeros except for those on the $n$th row and $v_{m,i}(n)$th column, which are set as $1, n=0,...,s_{m,i}-1$. Thus, the frequency sequence of $\boldsymbol{f}_m$ on the $i$th stopband can be expressed as $\boldsymbol{\bar{F}}_{m,i}=\mathbf{V}_{m,i}\mathbf{W}_{NP}\mathbf{P}\boldsymbol{f}_m$. The energy of the $i$th stopband is $\boldsymbol{\bar{F}}_{m,i}^H\boldsymbol{\bar{F}}_{m,i}$, which is subject to a constraint that it is less than a constant $e_{m,i}$. Therefore, the stopband energy constraints for user-$m$ are written as $\boldsymbol{f}_m^H\mathbf{E}_{m,i}\boldsymbol{f}_m\leq e_{m,i}, i=1,...,S_m$, where $\mathbf{E}_{m,i}\triangleq\mathbf{P}^T\mathbf{W}_{NP}^H\mathbf{V}_{m,i}^T
\mathbf{V}_{m,i}\mathbf{W}_{NP}\mathbf{P}$ is a positive definite conjugate symmetric matrix. Thus, the joint optimization problem with stopband energy constraints can be formulated as

\begin{subequations}
\begin{align}\label{optproblem5}
({\rm P2}):\mathop{\max}_{\boldsymbol{f}_m,\mathbf{C}_m,\forall m}
&R_{sum}\\\label{optproblem5.1}
\textit{s.t.}\quad &\frac{1}{NP}{\rm tr}(\mathbf{F}_m \mathbf{U} \mathbf{C}_m \mathbf{U}^T \mathbf{F}_m^H)=P_m,\\\label{optproblem5.2}
&\boldsymbol{f}_m^H\mathbf{E}_{m,i}\boldsymbol{f}_m\leq e_{m,i}, i=1,...,S_m,\\\label{optproblem5.3}
&\boldsymbol{f}_m^H\boldsymbol{f}_m=1, m=1,...,M,
\end{align}
\end{subequations}

As mentioned in Section \uppercase\expandafter{\romannumeral3}, optimizing the covariance matrices and filters on all subbands simultaneously is impracticable. Optimal $\mathbf{C}_m$ and $\boldsymbol{f}_m, m=1,...,M$, should be found iteratively until the sum rate converges. For user-$m$ in each iteration, its $\mathbf{C}_m$ and $\boldsymbol{f}_m$ are alternately optimized by regarding other users' as known and fixed, where the alternating optimization refers to the procedure that fixing $\mathbf{C}_m$ to optimize $\boldsymbol{f}_m$, and then fixing $\boldsymbol{f}_m$ to optimize $\mathbf{C}_m$. Therefore, in this procedure, Problem $({\rm P2})$ leads to two associated suboptimization problems which will be shown in the following subsections.

\subsection{Covariance Matrix Optimization Given $\boldsymbol{f}_m$}
In this subsection, we consider the algorithm of optimizing the covariance matrix $\mathbf{C}_m$ in each iteration with the assumption that $\boldsymbol{f}_m, m=1,...,M$, are known and fixed. Without loss of generality, we initialize $\mathbf{C}_m=PP_m\mathbf{I}_N, m=1,...,M$, to satisfy the constraint (\ref{optproblem5.2}). Thus, the first suboptimization problem is written as
\begin{align}\label{subproblem1}
({\rm P2-1}):\mathop{\max}_{\mathbf{C}_m}\quad
&R_{sum}\\
\nonumber
\textit{s.t.}\quad &(\ref{optproblem5.1}).
\end{align}
Similar to (\ref{optproblem3}), the objective function of $({\rm P2-1})$ can be rewritten as
\begin{align}
\log_2\left|\mathbf{\Phi}_m+\frac{1}{N_0} \mathbf{\Lambda}_{\boldsymbol{H}_m} \mathbf{\Lambda}_{\boldsymbol{F}_m} \mathbf{W}_{NP} \mathbf{U} \mathbf{C}_m \mathbf{U}^T \mathbf{W}_{NP}^T \mathbf{\Lambda}_{\boldsymbol{F}_m}^H \mathbf{\Lambda}_{\boldsymbol{H}_m}^H \right|.
\end{align}
Define $\mathbf{M}_m\triangleq\frac{1}{\sqrt{N_0}}\mathbf{\Phi}_m^{-1/2} \mathbf{\Lambda}_{\boldsymbol{H}_m} \mathbf{\Lambda}_{\boldsymbol{F}_m} \mathbf{W}_{NP} \mathbf{U}$, $\mathbf{N}_m\triangleq \mathbf{\Lambda}_{\boldsymbol{F}_m} \mathbf{W}_{NP} \mathbf{U}$, and their GSVD can be expressed as
\begin{align}\label{GSVD}
\mathbf{M}_m=\mathbf{V}_{\mathbf{M}_m}
\mathbf{\Lambda}_{\mathbf{M}_m}\mathbf{X}_m^H,\quad
\mathbf{N}_m=\mathbf{V}_{\mathbf{N}_m}
\mathbf{\Lambda}_{\mathbf{N}_m}\mathbf{X}_m^H,
\end{align}
where $\mathbf{V}_{\mathbf{M}_m}$, $\mathbf{V}_{\mathbf{N}_m}\in\mathbb{C}^{NP\times NP}$ are the left unitary matrices, $\mathbf{\Lambda}_{\mathbf{M}_m}$, $\mathbf{\Lambda}_{\mathbf{N}_m}\in\mathbb{C}^{NP\times N}$ are diagonal matrices, and $\mathbf{X}_m\in \mathbb{C}^{N\times N}$ is an invertible matrix whose columns are orthogonal. Therefore, Problem $({\rm P2-1})$ can be transformed into
\begin{subequations}
\begin{align}\label{subproblem12}
({\rm P2-1'}):\mathop{\max}_{\mathbf{S}_m}\quad
&\log_2\left|\mathbf{I}_{NP}+\mathbf{\Lambda}_{\mathbf{M}_m}
\mathbf{S}_m \mathbf{\Lambda}_{\mathbf{M}_m}^H \right|\\
\textit{s.t.}\quad &{\rm tr}(\mathbf{\Lambda}_{\mathbf{N}_m}\mathbf{S}_m
\mathbf{\Lambda}_{\mathbf{N}_m}^H)=NPP_m,
\end{align}
\end{subequations}
where matrix $\mathbf{S}_m\triangleq\mathbf{X}_m^H\mathbf{C}_m\mathbf{X}_m$ is the new optimization variable. According to Hadamard's inequality \cite{TCJT91}, the optimal $\mathbf{S}_m$ is a diagonal matrix and can be obtained by water-filling algorithm. Then, the optimal covariance matrix is constructed by $\mathbf{C}_m^*=(\mathbf{X}_m^H)^{-1}\mathbf{S}_m^*\mathbf{X}_m^{-1}$.

\subsection{Waveform Optimization Given $\mathbf{C}_m$}
After $\mathbf{C}_m$ is optimized, constraints (\ref{optproblem5.2}) and (\ref{optproblem5.3}) can not be combined together.
With the general known covariance matrix $\mathbf{C}_m$, the waveform optimization problem for $\boldsymbol{f}_m$ is written as
\begin{align}\label{optproblem6}
({\rm P2-2}):\mathop{\max}_{\boldsymbol{f}_m}\quad &R_{sum}\\
\nonumber
\textit{s.t.}\quad &(\ref{optproblem5.1}), (\ref{optproblem5.2}), (\ref{optproblem5.3}).
\end{align}
\begin{lemma} With any optimized covariance matrix $\mathbf{C}_m^*$, $\mathbf{Q}_m=\frac{1}{N_0} \mathbf{W}_{NP} \mathbf{U} \mathbf{C}_m^* \mathbf{U}^T \mathbf{W}_{NP}^H$ is still a block matrix consisting of $P\times P$ subblocks, and each subblock is a diagonal matrix.
\end{lemma}
\begin{proof}
See Appendix D.
\end{proof}

Thus, $\mathbf{Q}_m$ can be transformed as (\ref{Qm}) with $[\boldsymbol{q}_{m,n}]_{iN+n}$ redefined as $[\boldsymbol{q}_{m,n}]_{iN+n}=\sqrt{\frac{[\mathbf{W}_N\mathbf{C}_m\mathbf{W}_N^H]_{n,n}}{N_0}}$, $i=0,...,P-1$, and constraint (\ref{optproblem5.1}) can be simplified as
\begin{align}\label{con1reform1}
\nonumber
&\frac{1}{NP}{\rm tr}(\mathbf{F}_m \mathbf{U} \mathbf{C}_m^* \mathbf{U}^T \mathbf{F}_m^H)=\frac{N_0}{NP}{\rm tr}(\mathbf{\Lambda}_{\boldsymbol{F}_m} \mathbf{Q}_m \mathbf{\Lambda}_{\boldsymbol{F}_m}^H)=N_0{\rm tr}(\sum_{n=0}^{N-1} \mathbf{\Lambda}_{\boldsymbol{q}_{m,n}} \mathbf{W}_{NP} \mathbf{P} \boldsymbol{f}_{m} \boldsymbol{f}_{m}^H \mathbf{P}^T \mathbf{W}_{NP}^H \mathbf{\Lambda}_{\boldsymbol{q}_{m,n}}^H)\\
=&N_0\boldsymbol{f}_m^H \mathbf{P}^T \mathbf{W}_{NP}^H (\sum_{n=0}^{N-1} \mathbf{\Lambda}_{\boldsymbol{q}_{m,n}}^H \mathbf{\Lambda}_{\boldsymbol{q}_{m,n}}) \mathbf{W}_{NP} \mathbf{P} \boldsymbol{f}_{m}=N_0\boldsymbol{f}_m^H \mathbf{R}_m \boldsymbol{f}_m=P_m,
\end{align}
i.e.,
\begin{align}\label{optproblem5.4}
\boldsymbol{f}_m^H \mathbf{R}_m \boldsymbol{f}_m=\frac{P_m}{N_0},
\end{align}
where $\mathbf{R}_m\triangleq\mathbf{P}^T \mathbf{W}_{NP}^H (\sum_{n=0}^{N-1} \mathbf{\Lambda}_{\boldsymbol{q}_{m,n}}^H \mathbf{\Lambda}_{\boldsymbol{q}_{m,n}}) \mathbf{W}_{NP} \mathbf{P}$ is a Hermitian matrix. Besides, the objective function (\ref{optproblem6}) can be transformed to (\ref{optproblem2-sub}). Unfortunately, the manifold-based gradient ascent algorithm is not applicable any more, since constraints in Problem $({\rm P2-2})$ can not be regarded as a manifold. Here, we use second-order Taylor approximation near the result of the $t$th iteration to simplify the objective function in the $(t+1)$th iteration. To be more specific, $\bar{R}_m(\boldsymbol{f}_m^{(t+1)})$ can be approximated by
\begin{align}\label{Taylor2}
\nonumber
\bar{R}_m(\boldsymbol{f}_m^{(t+1)})
\approx&\bar{R}_m(\boldsymbol{f}_m^{(t)})+\mathfrak{R}[(\nabla_{\bar{R}_m}^{(t)})^H
(\boldsymbol{f}_m^{(t+1)}-\boldsymbol{f}_m^{(t)})]
+\frac{1}{2}(\boldsymbol{f}_m^{(t+1)}-\boldsymbol{f}_m^{(t)})^H
(\nabla_{\bar{R}_m}^2)^{(t)}(\boldsymbol{f}_m^{(t+1)}-\boldsymbol{f}_m^{(t)})\\
\nonumber
=&\frac{1}{2}(\boldsymbol{f}_m^{(t+1)})^H(\nabla_{\bar{R}_m}^2)^{(t)}\boldsymbol{f}_m^{(t+1)}
+\frac{1}{2}(\boldsymbol{\eta}_m^{(t+1)})^H\boldsymbol{f}_m^{(t+1)}
+\frac{1}{2}(\boldsymbol{f}_m^{(t+1)})^H\boldsymbol{\eta}_m^{(t+1)}\\
&-(\nabla_{\bar{R}_m}^{(t)})^H\boldsymbol{f}_m^{(t)}
-(\boldsymbol{f}_m^{(t)})^H\nabla_{\bar{R}_m}^{(t)}+\bar{R}_m(\boldsymbol{f}_m^{(t)}),
\end{align}
where $\nabla_{\bar{R}_m}^2$ denotes the second order derivative of $\bar{R}_m(\boldsymbol{f}_m)$, which is given by
\begin{align}\label{secder}
\nabla_{\bar{R}_m}^2=\frac{1}{\ln2}\sum_{n=0}^{N-1}\frac{(1+\boldsymbol{f}_m^H\mathbf{B}_{m,n}\boldsymbol{f}_m)
\mathbf{B}_{m,n}-\mathbf{B}_{m,n}\boldsymbol{f}_m\boldsymbol{f}_m^H\mathbf{B}_{m,n}}
{(1+\boldsymbol{f}_m^H\mathbf{B}_{m,n}\boldsymbol{f}_m)^2},
\end{align}
and $\boldsymbol{\eta}_m^{(t+1)}\triangleq \nabla_{\bar{R}_m}^{(t)}-(\nabla_{\bar{R}_m}^2)^{(t)}
\boldsymbol{f}_m^{(t)}$. Thus, Problem $({\rm P2-2})$ is rewritten as
\begin{align}\label{optproblem9}
({\rm P2-2'}):\mathop{\max}_{\boldsymbol{f}_m}\quad &\boldsymbol{f}_m^H(\nabla_{\bar{R}_m}^2)^{(t)}\boldsymbol{f}_m+
\boldsymbol{\eta}_m^H\boldsymbol{f}_m+\boldsymbol{f}_m^H\boldsymbol{\eta}_m\\
\nonumber
\textit{s.t.}\quad &(\ref{optproblem5.2}), (\ref{optproblem5.3}), (\ref{optproblem5.4}),
\end{align}
where the constant $\bar{R}_m(\boldsymbol{f}_m^{(t)})-(\nabla_{\bar{R}_m}^{(t)})^H\boldsymbol{f}_m^{(t)}
-(\boldsymbol{f}_m^{(t)})^H\nabla_{\bar{R}_m}^{(t)}$, the coefficient $\frac{1}{2}$ and the superscript $(t+1)$ in the objective function are omitted for clarity. Problem $({\rm P2-2'})$ is a non-convex quadratic constrained quadratic programming (QCQP) due to the equality constraints. To cope with this obstacle, we apply the semidefinite relaxation (SDR) \cite{LMSY10} technique to obtain an approximate solution. Specificly, let $\mathbf{\bar{F}}_m\triangleq\begin{bmatrix}\boldsymbol{f}_m\boldsymbol{f}_m^H & \boldsymbol{f}_mf^* \\ \boldsymbol{f}_m^Hf & \vert f\vert^2\end{bmatrix}$, where $f$ is an auxiliary variable \cite{WJAH10}. Then, $({\rm P2-2'})$ is equivalent to
\begin{subequations}
\begin{align}\label{optproblem10}
({\rm P2-2''}):\mathop{\max}_{\mathbf{\bar{F}}_m}\quad &{\rm tr}(\begin{bmatrix}(\nabla_{\bar{R}_m}^2)^{(t)} & \boldsymbol{\eta}_m \\ \boldsymbol{\eta}_m^H & 0 \end{bmatrix}\mathbf{\bar{F}}_m)\\
\textit{s.t.}\quad &{\rm tr}(\begin{bmatrix}\mathbf{E}_{m,i} & \mathbf{0}_{N_f\times 1} \\ \mathbf{0}_{1\times N_f} & 0 \end{bmatrix}\mathbf{\bar{F}}_m)\leq e_{m,i}, \forall i,\\
&{\rm tr}(\begin{bmatrix}\mathbf{R}_m & \mathbf{0}_{N_f\times 1} \\ \mathbf{0}_{1\times N_f} & 0 \end{bmatrix}\mathbf{\bar{F}}_m)=\frac{P_m}{N_0},\\
&{\rm tr}(\begin{bmatrix}\mathbf{I}_{N_f} & \mathbf{0}_{N_f\times 1} \\
\mathbf{0}_{1\times N_f} & 0 \end{bmatrix}\mathbf{\bar{F}}_m)=1,\\
&{\rm tr}(\begin{bmatrix}\mathbf{0}_{N_f\times N_f} & \mathbf{0}_{N_f\times 1} \\
\mathbf{0}_{1\times N_f} & 1 \end{bmatrix}\mathbf{\bar{F}}_m)=1,\\\label{optproblem10.6}
&{\rm rank}(\mathbf{\bar{F}}_m)=1,\\\label{optproblem10.7}
&\mathbf{\bar{F}}_m\succeq 0,
\end{align}
\end{subequations}
where constraint (\ref{optproblem10.7}) ensures that $\mathbf{\bar{F}}_m$ is positive semidefinite. It should be noted that constraint (\ref{optproblem10.6}) is hard to satisfy, and when it is removed, Problem $({\rm P2-2''})$ is a convex SDP and can be solved efficiently via existing software, e.g., CVX \cite{MGSB18}. Thus, we use CVX to solve Problem $({\rm P2-2''})$ without constraint (\ref{optproblem10.6}). Let $\mathbf{\bar{F}}_m^*$ denotes the optimal solution and suppose ${\rm rank}(\mathbf{\bar{F}}_m^*)=r$. The eigenvalue decomposition (EVD) of $\mathbf{\bar{F}}_m^*$ is expressed as $\mathbf{\bar{F}}_m^*=\sum_{i=1}^r\bar{\lambda}_i\boldsymbol{\bar{f}}_i\boldsymbol{\bar{f}}_i^H$, where $\bar{\lambda}_1\geq...\geq\bar{\lambda}_r$ are the eigenvalues and $\boldsymbol{\bar{f}}_1,...,\boldsymbol{\bar{f}}_r\in \mathbb{C}^{N_f+1}$ are the respective eigenvectors. Since the best rank-one approximation $\mathbf{\hat{F}}_m^*$ to $\mathbf{\bar{F}}_m^*$ (in the least two-norm sense) is given by $\mathbf{\hat{F}}_m^*=\bar{\lambda}_1\boldsymbol{\bar{f}}_1\boldsymbol{\bar{f}}_1^H$, we take the first $N_f$ elements of $\sqrt{\bar{\lambda}_1}\boldsymbol{\bar{f}}_1$ as the optimal $\boldsymbol{f}_m$.

Taking both covariance matrix and waveform optimization into account, the whole sum rate maximization algorithm for solving Problem $({\rm P2})$ is summarized in Algorithm \ref{whole1}.
\begin{algorithm}
	\caption{Sum Rate Maximization Algorithm for Solving $({\rm P2})$}
    \renewcommand{\algorithmicrequire}{\textbf{Input:}}
    \renewcommand{\algorithmicensure}{\textbf{Output:}}
	\label{whole1}
	\begin{algorithmic}[1]
        \REQUIRE $M$, $N$, $P$, $N_0$, $P_m$, $\boldsymbol{h}_m,m=1,...,M$
        \ENSURE Optimal $\mathbf{C}_m$ and $\boldsymbol{f}_m, m=1,...,M$
		\STATE Initialize $\boldsymbol{f}_m$, $\mathbf{C}_m=PP_m\mathbf{I}_N$ and construct $\mathbf{U}$, $\mathbf{H}_m$ and $\mathbf{F}_m, m=1,...,M$.
		\REPEAT
        \FOR{$m=1$ to $M$}
        \STATE Compute $\nabla_{\bar{R}_m}^2$ and $\nabla_{\bar{R}_m}$ by (\ref{secder}) and (\ref{nabla}), respectively.
		\STATE Solve Problem $({\rm P2-2''})$ via CVX and obtain the optimal solution $\mathbf{\bar{F}}_m^*$.
        \STATE Use EVD for $\mathbf{\bar{F}}_m^*$ and sort the eigenvalues in non-increasing order $\bar{\lambda}_1\geq...\geq\bar{\lambda}_r$.
        \STATE Set $\boldsymbol{f}_m$ as the first $N_f$ elements of $\sqrt{\bar{\lambda}_1}\boldsymbol{\bar{f}}_1$.
        \STATE Compute $\mathbf{M}_m=\frac{1}{\sqrt{N_0}}\mathbf{\Phi}_m^{-1/2} \mathbf{\Lambda}_{\boldsymbol{H}_m} \mathbf{\Lambda}_{\boldsymbol{F}_m} \mathbf{W}_{NP} \mathbf{U}$ and $\mathbf{N}_m=\mathbf{\Lambda}_{\boldsymbol{F}_m} \mathbf{W}_{NP} \mathbf{U}$.
        \STATE Use GSVD for $\mathbf{M}_m$ and $\mathbf{N}_m$ to obtain $\mathbf{\Lambda}_{\mathbf{M}_m}$, $\mathbf{\Lambda}_{\mathbf{N}_m}$ and $\mathbf{X}_m$.
        \STATE Use water-filling algorithm to obtain the diagonal matrix $\mathbf{S}_m$.
        \STATE Set $\mathbf{C}_m=(\mathbf{X}_m^H)^{-1}\mathbf{S}_m\mathbf{X}_m^{-1}$.
        \ENDFOR
        \UNTIL the sum rate converges.
	\end{algorithmic}
\end{algorithm}

\subsection{Complexity Analysis}
In this subsection we analyze the complexity of Algorithm \ref{whole1}. In a single iteration, the dominated computational cost comes from Step 4, Step 5 and Step 9. For Step 4, its complexity is the same as that of Step 4 in Algorithm \ref{waveopt1}, which is $\mathcal{O}(MNP^2N_f+MNPN_f^2)$. For Step 5, it involves the computations of matrix $\mathbf{R}_m$ and $\mathbf{E}_{m,i}$, as well as using interior point algorithm for solving SDP in CVX. The complexity of both the matrix computations is $\mathcal{O}(MNPN_f^2)$ (this is the worst case for $\mathbf{E}_{m,i}$), and that of interior point algorithm is $\mathcal{O}(T_{cvx}(S_mN_f^3+S_m^2N_f^2))$ \cite{BBJY05} in the worst case, where $T_{cvx}$ is the total iteration number that CVX needs. The value of $T_{cvx}$ is determined by the software, and we observe in our simulations that it is about $15$. For Step 9, considering that GSVD involves QR decomposition \cite{GHCV96}, its complexity is $\mathcal{O}(MN^3P)$. Thus, the normalized complexity of the proposed algorithm per QAM symbol is $\mathcal{O}(T_{iter}(P^2N_f+PN_f^2+T_{cvx}(S_mN_f^3+S_m^2N_f^2)/N+N^2P))$, where $T_{iter}$ is the total number of iteration that the algorithm needs and $S_1=...=S_M$ is assumed.

\section{Simulation Results}
In this section we provide simulation results to show the performance of the proposed waveform and covariance matrix optimization algorithm. Different configurations including various lengths of filters and upsampling factors are simulated to expose the performance of the algorithms. The channel is modeled as a $L_h$-tap complex random Gaussian channel with equal power gains among the taps and the total power is normalized to unity. Here, $L_h=10$. Other key configurations and parameters are given as follows. We assume there are totally $M=8$ users, which indicates the number of subbands is $8$. The number of symbols per subband is $N=48$. Besides, the length of CP is $\lceil (N_f+L_h-1)/P\rceil$, which is the minimum integer no smaller than $(N_f+L_h-1)/P$. All users share the same transmit signal-to-noise ratio (SNR), which is denoted by ${\gamma}=\frac{P_m}{N_0}$. For upsampling factors, we select three cases: for $P=1$, it is maximally oversampled; for $P=M/4$, it is $4$x oversampled; for $P=M$, it is critically sampled.

\subsection{Performance of Waveform Optimization on The Full Band}
In Fig. \ref{f_conver}, we aim to explicitly show the convergence behavior of Algorithm \ref{waveopt1} and the effect of filter length on it, where we take three different filter lengths ($N_f=16/32/48$) into consideration. Fig. \ref{f_conver} shows the change of the sum rate with the number of outer iterations in Algorithm \ref{waveopt1} when $P=M$ and ${\gamma}=10$ dB. We can observe that the sum rate increases monotonically, and the proposed waveform optimization algorithm converges dramatically fast (about 5 outer iterations). Besides, the filter length has negligible impact on the speed of convergence.

\begin{figure}[htbp]
\begin{minipage}[t]{0.5\linewidth}
\setlength{\abovecaptionskip}{-0.3cm}
\centering
\includegraphics[width=3in]{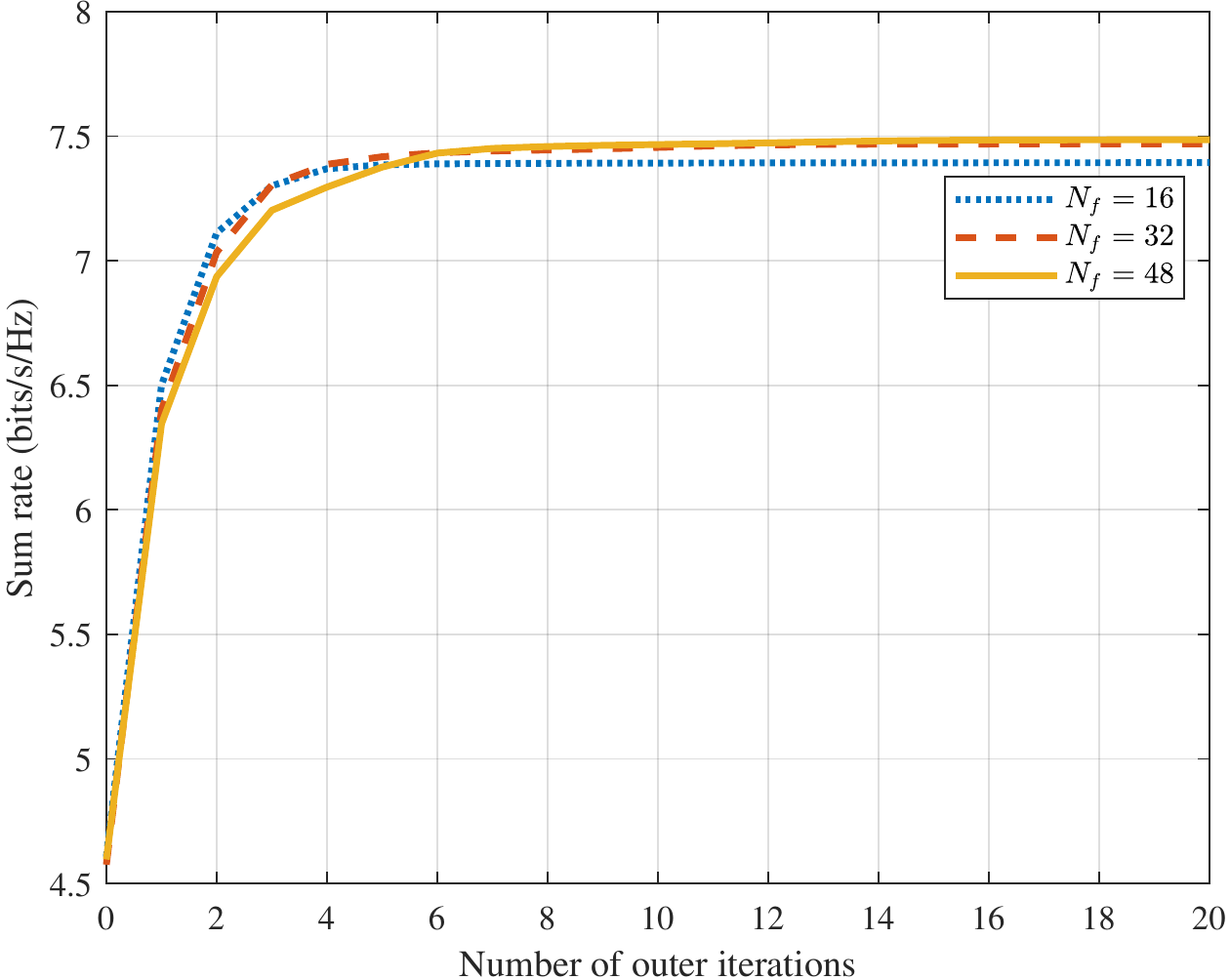}
\caption{Convergence behavior of Algorithm \ref{waveopt1}.}
\label{f_conver}
\end{minipage}%
\hfill
\begin{minipage}[t]{0.5\linewidth}
\setlength{\abovecaptionskip}{-0.3cm}
\centering
\includegraphics[width=3in]{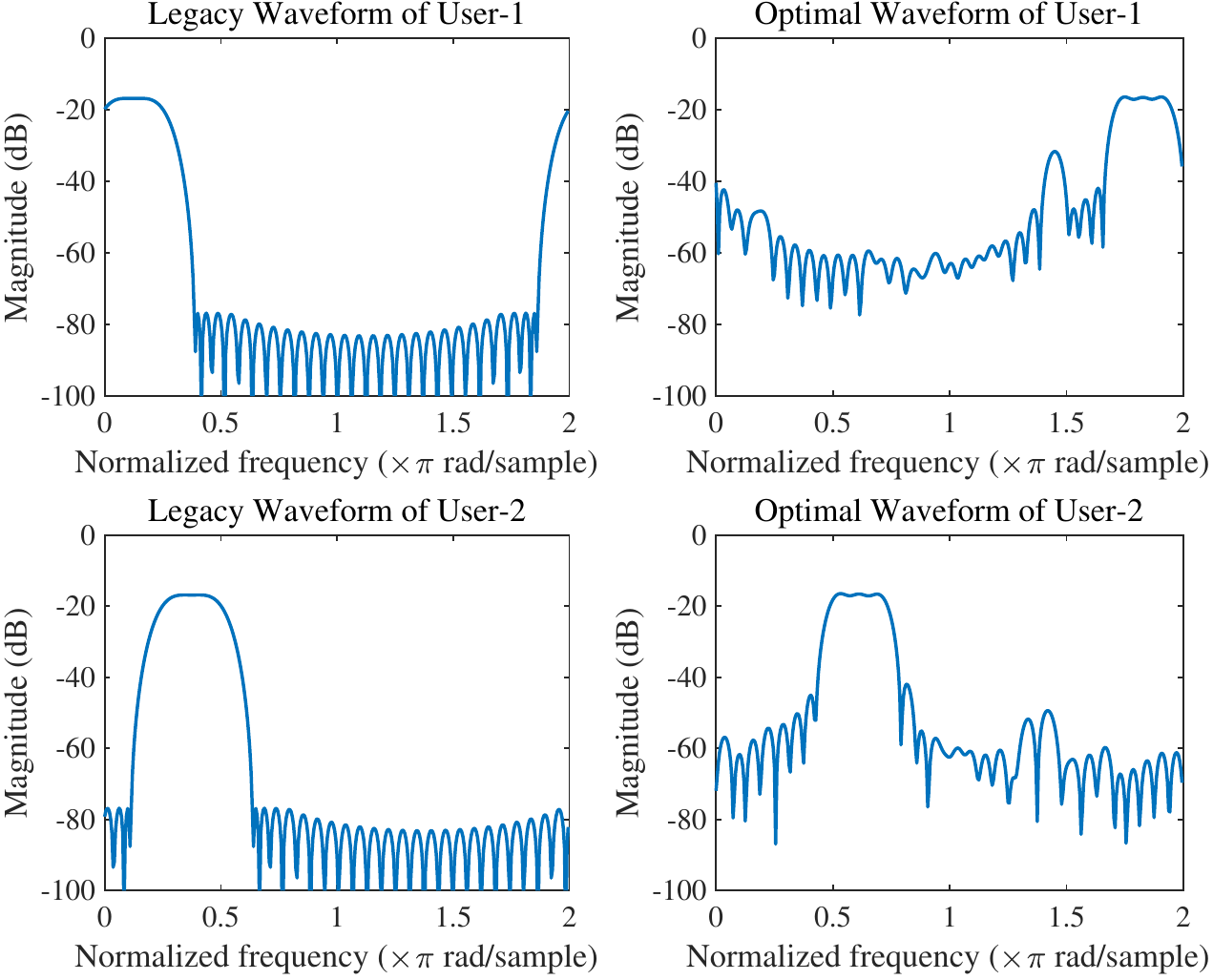}
\caption{Legacy and optimal waveforms.}
\label{f_wave}
\end{minipage}
\end{figure}

Next, we compare the performance of the following two systems which share the same framework as shown in Fig. \ref{system model}, and the difference is only in the design method of the filters:

1) \textbf{CP-FBMA with legacy waveform}: filters are borrowed directly from the ones in EMFB without any optimization, and each filter occupies a wide band rather than a narrow band. The specific filter design algorithm can be referred to \cite{DZWW17,DGZW17};

2) \textbf{CP-FBMA with proposed optimal waveform}: filters are optimized based on Algorithm \ref{waveopt1} on the full band.

Fig. \ref{f_wave} shows both legacy and optimal waveforms of the first two users in frequency domain. Key parameters are set as $N_f=32$, $P=M$ and ${\gamma}=10$ dB. It can be observed that legacy waveforms on different subbands share the same waveform shape, given that all synthesis filters are derived from a prototype filter, while such regularity of waveforms in frequency domain is broken after optimization, and positions of the passbands are also changed to better match the frequency response of the channels. Besides, the sum rate is improved by $63.26\%$ in this case.

\begin{figure}[htbp]
\begin{minipage}[t]{0.5\linewidth}
\setlength{\abovecaptionskip}{-0.3cm}
\centering
\includegraphics[width=3in]{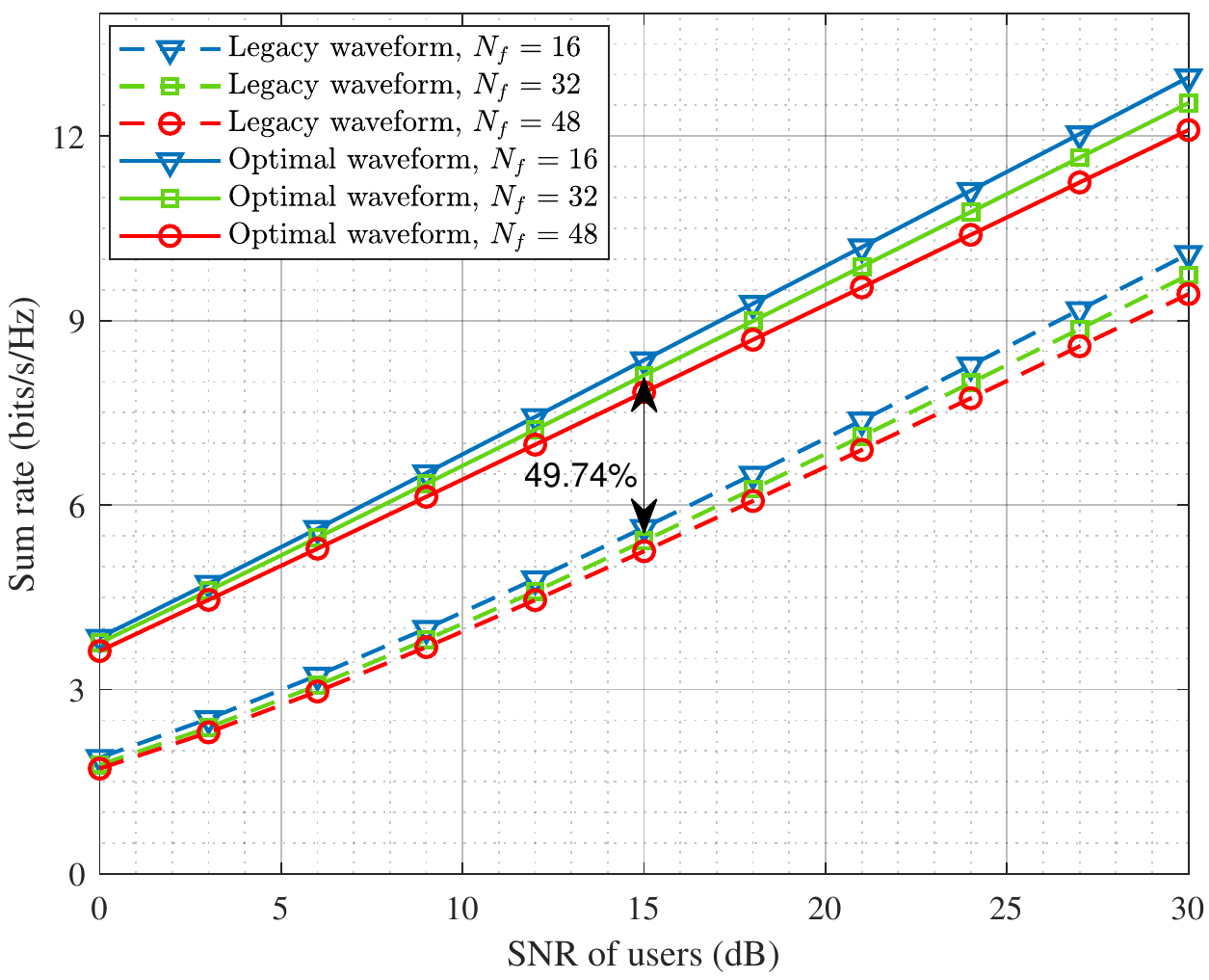}
\caption{Performance comparison under $P=M$.}
\label{f_rate_Nf}
\end{minipage}%
\begin{minipage}[t]{0.5\linewidth}
\setlength{\abovecaptionskip}{-0.3cm}
\centering
\includegraphics[width=3in]{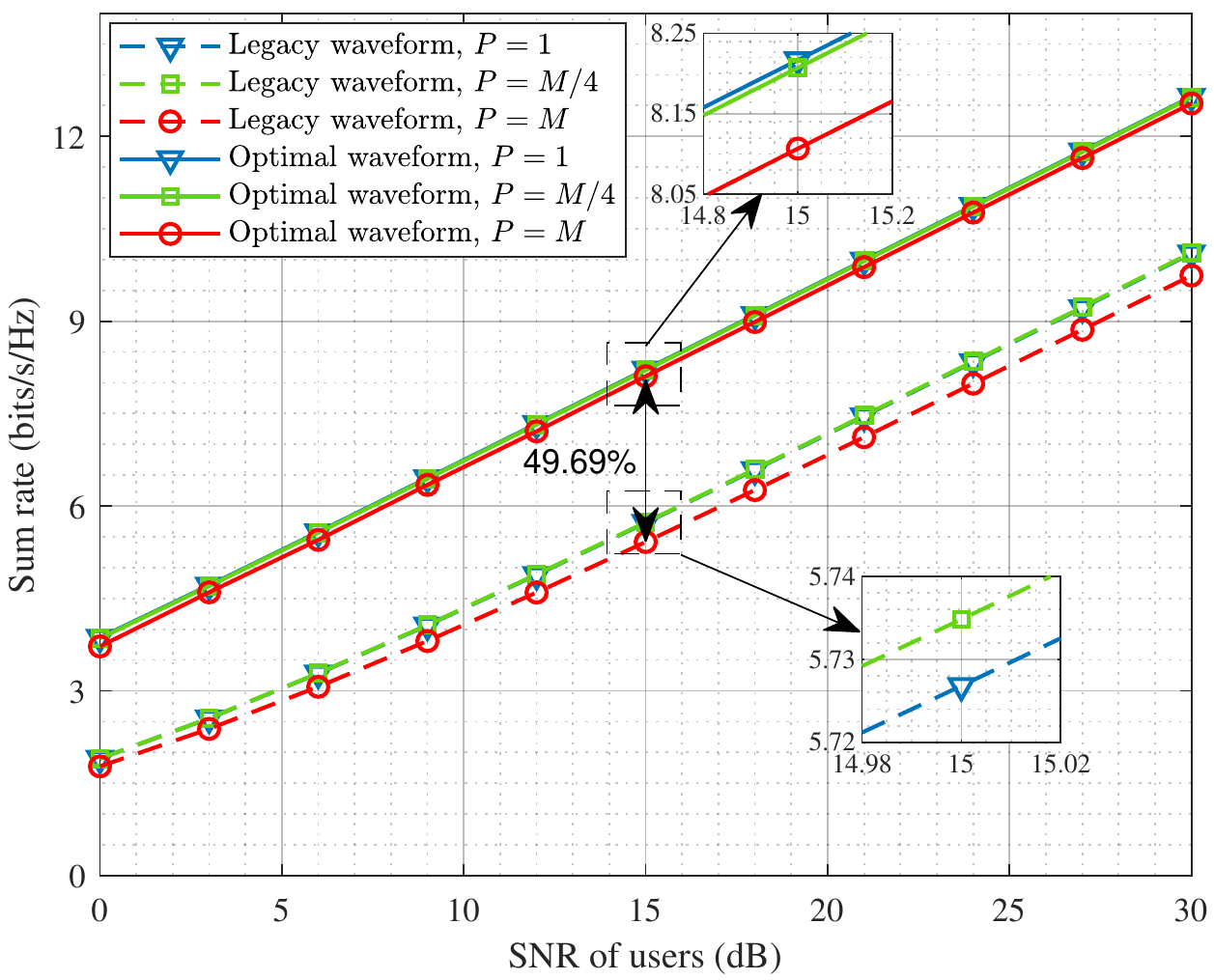}
\caption{Performance comparison under $N_f=32$.}
\label{f_rate_P}
\end{minipage}
\end{figure}

Fig. \ref{f_rate_Nf} and Fig. \ref{f_rate_P} demonstrate the sum rate of the system with optimal waveform compared to that with legacy waveform, where ${\gamma}$ ranges from $0$ to $30$ dB and the impacts of filter length and upsampling factor are considered. We can see that: 1) by using the proposed waveform optimization algorithm, the sum rate of the system is improved significantly. Take the case of $N_f=32$ and ${\gamma}=15$ dB in Fig. \ref{f_rate_Nf} as an example. The sum rate is increased by $49.74\%$ compared to that of the system with legacy waveform; 2) filter length has no noticeable impact on the sum rate, with or without the optimal waveform. Nevertheless, for the system with optimal waveform, even in the case that the filter length is shortest ($N_f=16$), its sum rate is still higher than that of the system with legacy waveform whose filter length is much longer. This indicates that in non-orthogonal CP-FBMA system, by using optimal waveform, we can significantly increase the rate while even reducing the filter length, which is beneficial to reduce the complexity of the optimization algorithm and transmission delay of per QAM symbol; 3) although decreasing upsampling factor $P$ can obtain a visible performance gain of the sum rate for legacy waveform, it is negligible compared to that by waveform optimization. For instance, the sum rate is only increased by $5.89\%$ via changing $P$ from $M$ to $M/4$ when $\gamma=15$ dB in Fig. \ref{f_rate_P}. However, using the optimization algorithm, it is increased by $49.69\%$ under $P=M$.

Finally, we compare the performance of the above two non-orthogonal CP-FBMA systems with EMFB which is selected as a representative of the conventional orthogonal FBMA. It is employed in multi-user communication where every subcarrier is exclusively allocated to a specific user for access. The configurations are set as follows: the number of users is $M=8$. For EMFB (denoted as conventional FBMA in the legend), the number of subcarriers is $M'=8/64/128$ (each user occupies $1/8/16$ subcarriers), the length of filter is $N_f=4M'$, the upsampling factor $P=M'/2$, the input symbols are real and no CP is inserted; For CP-FBMA, the length of filter is $N_f=4M$ and the upsampling factor $P=M$. Fig. \ref{rate_FBMA} shows the sum rate of the three systems. It can be observed that the rate performance of conventional FBMA is almost the same as that of CP-FBMA with legacy waveform regardless of $M'$. But it tends to perform worse in low SNR regime and better in high SNR regime. Obviously, by using the proposed optimization algorithm in CP-FBMA system, the sum rate can be greatly improved, whereas conventional FBMA can not apply the waveform optimization algorithm, since it needs to meet orthogonality between subbands in real domain.

\begin{figure}[ht]
\setlength{\abovecaptionskip}{-0.3cm}
\centering{\includegraphics[width=3in]{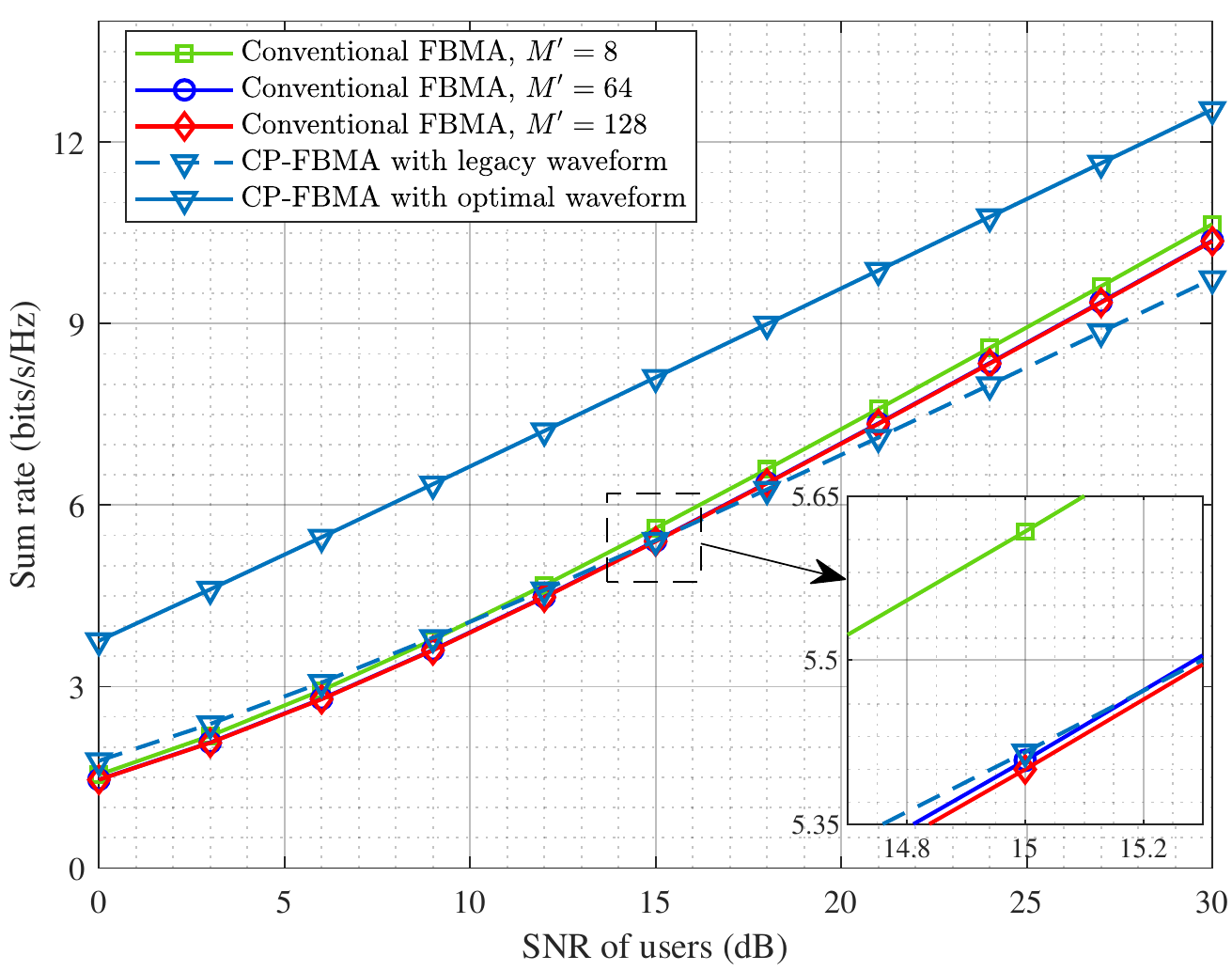}}
\caption{Performance comparison between conventional FBMA and CP-FBMA.}
\label{rate_FBMA}
\end{figure}

\subsection{Performance of Joint Optimization}
In this subsection, we show the performance of joint optimization, where stopband energy constraints are taken into account to meet requirements of the allowed fragmented frequency bands in practice. Given constraints on stopband, a direct filter design method is to use some classical filter design methods. For example, the complex and nonlinear-phase equiripple finite impulse response (FIR) filter design \cite{LJJH95} which is packaged in $cfirpm$ function in MATLAB. However, only positions and amplitudes of stopbands are allowed to be specified when using $cfirpm$ function, not the energy. Thus, after generating the filters, we compute the stopband energy of each filter $\bar{e}_{m,i}, m=1,...,M, i=1,...,S_m$, and set them as the upperbounds of stopband energy in optimization, which ensures a fair comparison between waveforms with or without optimization. Fig. \ref{C_f_wave_incon} shows waveforms designed by $cfirpm$ function and optimal waveforms based on the proposed joint optimization algorithm in frequency domain, where the shadow areas denote stopbands and the maximum allowed energy of each stopband of the first two users is set as $\bar{e}_{1,1}=3.23\times 10^{-4}$, $\bar{e}_{1,2}=2.24\times 10^{-4}$, $\bar{e}_{2,1}=3.53\times 10^{-4}$, $\bar{e}_{2,2}=3.52\times 10^{-4}$. Other configurations are $N_f=32$, $P=M$ and ${\gamma}=10$ dB. It can be observed that: 1) all the waveforms satisfy their respective stopband energy constraints; 2) compared to the waveforms designed by $cfirpm$ function, the optimal waveforms are irregular.

\begin{figure}[htbp]
\begin{minipage}[t]{0.5\linewidth}
\setlength{\abovecaptionskip}{-0.3cm}
\centering
\includegraphics[width=3in]{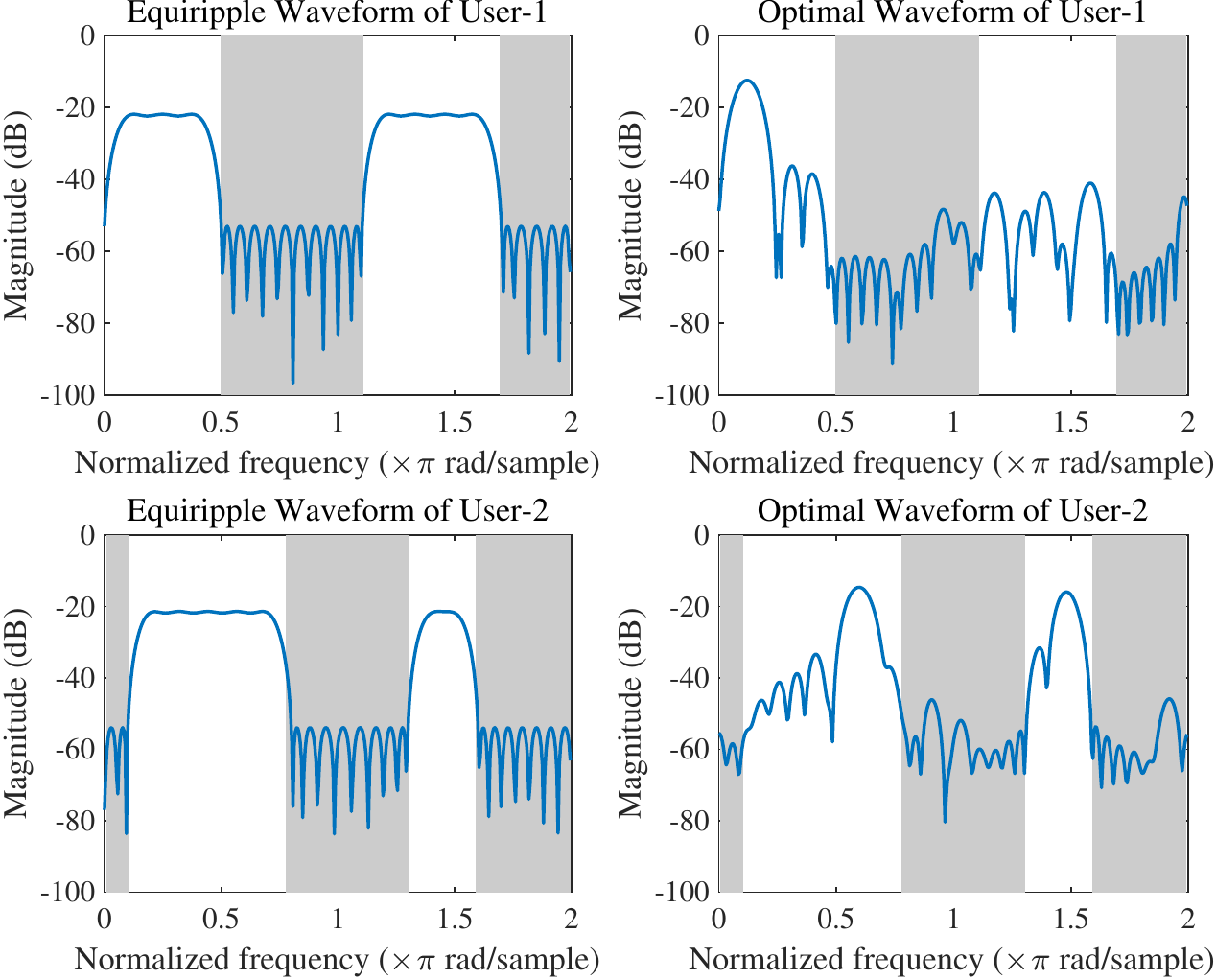}
\caption{Equiripple and optimal waveforms.}
\label{C_f_wave_incon}
\end{minipage}%
\hfill
\begin{minipage}[t]{0.5\linewidth}
\setlength{\abovecaptionskip}{-0.3cm}
\centering
\includegraphics[width=3in]{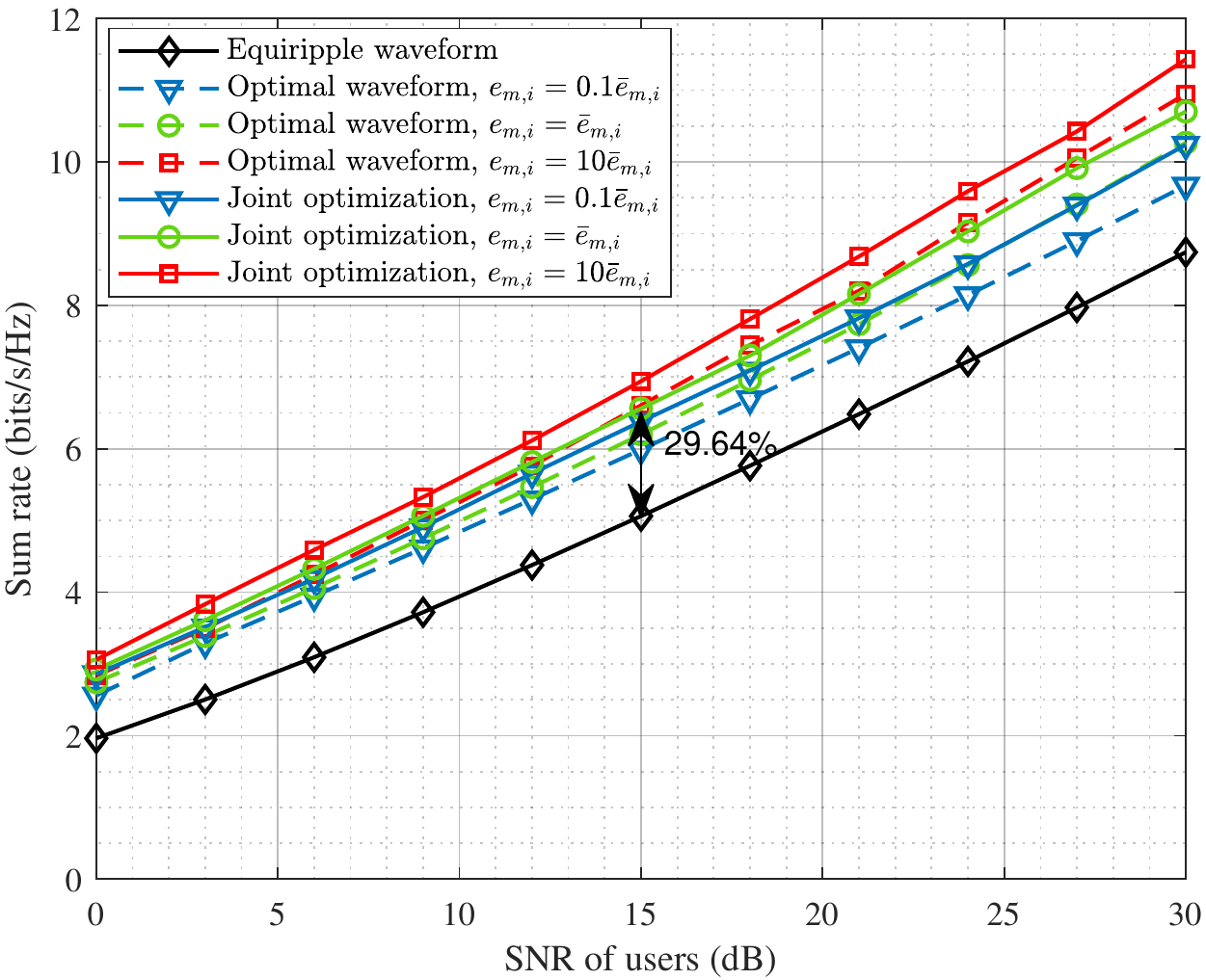}
\caption{Performance comparison under $N_f=32$ and $P=M$.}
\label{C_f_rate_incon}
\end{minipage}
\end{figure}

In Fig. \ref{C_f_rate_incon}, we compare the performance of the following systems with the same stopband energy constraints:

1) \textbf{CP-FBMA with equiripple waveform}: filters are designed by $cfirpm$ function with $\mathbf{C}_m=PP_m\mathbf{I}_N, m=1,...,M$. Besides, each filter occupies all available bandwidth with stopband energy constraints on the respective subband;

2) \textbf{CP-FBMA with optimal waveform only}: filters are optimized by Algorithm \ref{whole1} while fixing $\mathbf{C}_m=PP_m\mathbf{I}_N, m=1,...,M$, in each iteration;

3) \textbf{CP-FBMA with joint optimization}: both filters and covariance matrices are optimized by Algorithm \ref{whole1}.

As illustrated in the previous subsection, changing $N_f$ or $P$ has no apparent effect on the performance of the optimization, Thus, they are fixed as $N_f=32$ and $P=M$. We can observe from Fig. \ref{C_f_rate_incon} that: 1) though there are more constraints on filters, the joint optimization algorithm can still improve the sum rate of non-orthogonal CP-FBMA system significantly. Take the case of $\gamma=15$ dB and $e_{m,i}=\bar{e}_{m,i}$ as an example. The sum rate is increased by $29.64\%$ compare to that of the system with equiripple waveform; 2) system performance improves with the increasing of the stopband energy $e_{m,i}, \forall m,i$. However, even they are set as tiny values, the improvement of the sum rate brought by the joint optimization algorithm is still considerable; 3) the gain of joint optimization over waveform-only optimization is pronounced. This indicates that covariance matrix optimization plays an important role in improving the sum rate of the system.

\section{Conclusions and Future Works}
This paper firstly presents the waveform optimization algorithm for non-orthogonal CP-FBMA to maximize the sum rate. Then, covariance matrix optimization and stopband energy constraints are also considered to further improve the system performance and satisfy the requirement of frequency band utilization in practical applications. Through simulation results, we can see that the optimization algorithms converge dramatically fast and can significantly improve the sum rate of the system. Moreover, the sum rate of the optimal waveform with shorter filter length is apparently higher than that of the legacy waveform with longer filter length. Finally, by using the proposed joint optimization algorithm, the rate performance can still be improved even with stopband energy constraints. As a preliminary work in waveform design for non-orthogonal FBMA framework, this paper aims to reveal the maximum potential rate gain of CP-FBMA through waveform optimization. It is based on instantaneous CSI with a single antenna at the receiver side. To make the system more realistic, it is worth studying waveform optimization based on statistical CSI and with multiple, or even massive receive antennas in the future. Waveform optimization in the downlink scenario is also among other potential works.

\appendices
\section{Proof of Lemma 1}
By setting $PP_m\mathbf{I}_N=\mathbf{C}_m$ in $\frac{1}{NP}{\rm tr}(\mathbf{F}_m \mathbf{U} \mathbf{C}_m \mathbf{U}^H \mathbf{F}_m^H)=P_m$, we only need to prove when $\boldsymbol{f}_m^H\boldsymbol{f}_m=1$, ${\rm tr}(\mathbf{F}_m\mathbf{U}\mathbf{U}^H\mathbf{F}_m^H)=N$ holds. As mentioned that $N_f$ is a few times of $M$ in Section \uppercase\expandafter{\romannumeral2}, we suppose $N_f=\delta M$, where $\delta$ is a positive integer. Note that $\mathbf{F}_m\mathbf{U}$ can be divided into $N$ subblocks by every $P$ rows, i.e., $\mathbf{F}_m\mathbf{U}=[\mathbf{J}_{m,0}^T,...,\mathbf{J}_{m,N-1}^T]^T$.
Each subblock $\mathbf{J}_{m,n}\in \mathbb{C}^{P\times N}$ contains the same nonzero elements at the same rows, which are written as
\begin{align}\label{subblock}
\begin{bmatrix}
[\boldsymbol{f}_m]_0   & [\boldsymbol{f}_m]_P    & \cdots & [\boldsymbol{f}_m]_{(\delta-1)P}   \\
[\boldsymbol{f}_m]_1   & [\boldsymbol{f}_m]_{P+1}  & \cdots & [\boldsymbol{f}_m]_{(\delta-1)P+1} \\
\vdots   & \vdots    & \ddots & \vdots        \\
[\boldsymbol{f}_m]_{P-1} & [\boldsymbol{f}_m]_{2P-1} & \cdots & [\boldsymbol{f}_m]_{\delta P-1}     \\
\end{bmatrix}.
\end{align}
Thus, ${\rm tr}(\mathbf{F}_m\mathbf{U}\mathbf{U}^H\mathbf{F}_m^H)$ can be further simplified as
\begin{align}
{\rm tr}(\mathbf{F}_m\mathbf{U}\mathbf{U}^H\mathbf{F}_m^H)
=\sum_{n=0}^{N-1}{\rm tr}(\mathbf{J}_{m,n}\mathbf{J}_{m,n}^H)
=N\sum_{j=0}^{P-1}\sum_{i=0}^{\delta-1}\vert [\boldsymbol{f}_m]_{iP+j}\vert^2
=N\boldsymbol{f}_m^H\boldsymbol{f}_m=N,
\end{align}
which completes the proof.

\section{Proof of Lemma 2}
First, we aim to prove that $\mathbf{\Phi}_m^{-1}$ has such matrix structure, and use hypothesis-deduction method to solve it. Define a $2\times 2$ Hermitian block matrix $\mathbf{\Phi}\triangleq\begin{bmatrix} \mathbf{\Lambda}_{1,1} & \mathbf{\Lambda}_{1,2} \\ \mathbf{\Lambda}_{1,2}^H & \mathbf{\Lambda}_{2,2} \end{bmatrix}$, where $\mathbf{\Lambda}_{p,q}$ is a full rank diagonal matrix of size $L\times L$. Without loss of generality, $L$ can be determined as any positive integer. Besides, let $\mathbf{\Phi}^{-1}\triangleq\begin{bmatrix} \mathbf{\Theta}_{1,1} & \mathbf{\Theta}_{1,2} \\ \mathbf{\Theta}_{2,1} & \mathbf{\Theta}_{2,2} \end{bmatrix}$ denotes the inverse matrix of $\mathbf{\Phi}$. Due to that $\mathbf{\Phi}\mathbf{\Phi}^{-1}=\mathbf{I}_{2L}$, we have
\begin{equation} \label{22equation}
\left\{ \begin{matrix}
\mathbf{\Lambda}_{1,1}\mathbf{\Theta}_{1,1}+\mathbf{\Lambda}_{1,2}\mathbf{\Theta}_{2,1}=\mathbf{I}_L,\quad
&\mathbf{\Lambda}_{1,1}\mathbf{\Theta}_{1,2}+\mathbf{\Lambda}_{1,2}\mathbf{\Theta}_{2,2}=\mathbf{0}_{L\times L}\hfill\\
\mathbf{\Lambda}_{1,2}^H\mathbf{\Theta}_{1,1}+\mathbf{\Lambda}_{2,2}\mathbf{\Theta}_{2,1}=\mathbf{0}_{L\times L},\quad
&\mathbf{\Lambda}_{1,2}^H\mathbf{\Theta}_{1,2}+\mathbf{\Lambda}_{2,2}\mathbf{\Theta}_{2,2}=\mathbf{I}_L\hfill.
\end{matrix}
\right.
\end{equation}
By solving (\ref{22equation}), we obtain
\begin{equation} \label{22result}
\left\{ \begin{matrix}
\mathbf{\Theta}_{1,1}=(\mathbf{\Lambda}_{1,1}-\mathbf{\Lambda}_{1,2}\mathbf{\Lambda}_{2,2}^{-1}\mathbf{\Lambda}_{1,2}^H)^{-1},\quad
&\mathbf{\Theta}_{1,2}=(\mathbf{\Lambda}_{1,2}^H-\mathbf{\Lambda}_{2,2}\mathbf{\Lambda}_{1,2}^{-1}\mathbf{\Lambda}_{1,1})^{-1}\hfill\\
\mathbf{\Theta}_{2,1}=[\mathbf{\Lambda}_{1,2}-\mathbf{\Lambda}_{1,1}(\mathbf{\Lambda}_{1,2}^H)^{-1}\mathbf{\Lambda}_{2,2}]^{-1},\quad
&\mathbf{\Theta}_{2,2}=(\mathbf{\Lambda}_{2,2}-\mathbf{\Lambda}_{1,2}^H\mathbf{\Lambda}_{1,1}^{-1}\mathbf{\Lambda}_{1,2})^{-1}\hfill.
\end{matrix}
\right.
\end{equation}
It is observed that $\forall p, q$, $\mathbf{\Theta}_{p,q}$ is a diagonal matrix. Moreover, since $\mathbf{\Phi}$ is a Hermitian matrix, for any $p$, $\mathbf{\Lambda}_{p,p}$ is a real matrix. Therefore, $\mathbf{\Theta}_{2,1}=\mathbf{\Theta}_{1,2}^H$, and $\mathbf{\Theta}_{1,1}$, $\mathbf{\Theta}_{2,2}$ are also real matrices, which completes the proof in the case that $\mathbf{\Phi}$ is a $2\times 2$ Hermitian block matrix.

Now, assume that $\mathbf{\Phi}$ has $R\times R$ subblocks, where $R>2$ is an arbitrary positive integer, and in this case, Lemma 2 is valid. We need to prove that when $\mathbf{\Phi}$ has $(R+1)\times (R+1)$ subblocks, Lemma 2 still holds, which is equivalent to that $\mathbf{\Theta}_{p,R+1}, \mathbf{\Theta}_{R+1,p}, p=1,...,R+1$, are diagonal matrices and satisfy $\mathbf{\Theta}_{p,R+1}=\mathbf{\Theta}_{R+1,p}^H$. Similar to the above, due to $\mathbf{\Phi}\mathbf{\Phi}^{-1}=\mathbf{\Phi}^{-1}\mathbf{\Phi}=\mathbf{I}_{RL}$, we have
\begin{equation} \label{RRequation1}
\sum_{q=1}^{R+1}\mathbf{\Lambda}_{p,q}\mathbf{\Theta}_{q,p}=
\left\{ \begin{matrix}
\mathbf{I}_L,&p=1\hfill\\
\mathbf{0}_{L\times L},&p=2,...,R+1\hfill,
\end{matrix}
\right.\quad
\sum_{q=1}^{R+1}\mathbf{\Theta}_{p,q}\mathbf{\Lambda}_{q,p}=
\left\{ \begin{matrix}
\mathbf{I}_L,&p=1\hfill\\
\mathbf{0}_{L\times L},&p=2,...,R+1\hfill.
\end{matrix}
\right.
\end{equation}
Thus, we obtain that
\begin{equation} \label{RRresult1}
\mathbf{\Theta}_{R+1,p}=
\left\{ \begin{matrix}
\mathbf{\Lambda}_{p,R+1}^{-1}(\mathbf{I}_L-\sum_{q=1}^{R}\mathbf{\Lambda}_{p,q}\mathbf{\Theta}_{q,p}),&p=1\hfill\\
-\mathbf{\Lambda}_{p,R+1}^{-1}\sum_{q=1}^{R}\mathbf{\Lambda}_{p,q}\mathbf{\Theta}_{q,p},&p=2,...,R+1\hfill,
\end{matrix}
\right.
\end{equation}
\begin{equation} \label{RRresult2}
\mathbf{\Theta}_{p,R+1}=
\left\{ \begin{matrix}
(\mathbf{I}_L-\sum_{q=1}^{R}\mathbf{\Theta}_{p,q}\mathbf{\Lambda}_{q,p})\mathbf{\Lambda}_{R+1,p}^{-1},&p=1\hfill\\
-(\sum_{q=1}^{R}\mathbf{\Theta}_{p,q}\mathbf{\Lambda}_{q,p})\mathbf{\Lambda}_{R+1,p}^{-1},&p=2,...,R+1\hfill
\end{matrix}
\right.
\end{equation}
are all diagonal matrices. Note that $\mathbf{\Lambda}_{p,q}=\mathbf{\Lambda}_{q,p}^H$, $\mathbf{\Theta}_{p,q}=\mathbf{\Theta}_{q,p}^H, p, q=1,...,R$, and $\mathbf{\Lambda}_{p,R+1}=\mathbf{\Lambda}_{R+1,p}^H, p=1,...,R+1$. Equations $\mathbf{\Theta}_{p,R+1}=\mathbf{\Theta}_{R+1,p}^H, p=1,...,R+1$, hold, which completes the hypothesis-deduction method.

Second, it can be easily verified that $\mathbf{\Phi}\mathbf{\Phi}$ has the same matrix structure with $\mathbf{\Phi}$, which means that $\mathbf{\Phi}^{1/2}$ also has such matrix structure. Besides, we can use the method of undetermined coefficients to find out the specific values of $\mathbf{\Phi}^{1/2}$. In summary, $\mathbf{\Phi}_m^{-1/2}$ has the same matrix structure with $\mathbf{\Phi}_m$, and the proof is completed.

\section{Proof of Proposition 1}
Since ${\rm grad}_{\bar{R}_m}$ is regarded as the projection of ${\nabla}_{\bar{R}_m}$ on the tangent space $T_{\boldsymbol{f}_m}(\mathcal{M})$ at $\boldsymbol{f}_m$, it can be expressed by ${\rm grad}_{\bar{R}_m}=\mathop{\rm argmin}_{
\boldsymbol{\zeta}\in T_{\boldsymbol{f}_m}\mathcal{M}}
{\Vert\boldsymbol{\zeta}-{\nabla}_{\bar{R}_m}\Vert}^2$. Using the standard Lagrangian multiplier method
\begin{align}\label{Lagrangian}
L(\boldsymbol{\zeta},\lambda)&=(\boldsymbol{\zeta}-{\nabla}_{\bar{R}_m})^H(\boldsymbol{\zeta}-{\nabla}_{\bar{R}_m}) -\lambda(\boldsymbol{\zeta}^H\boldsymbol{f}_m+\boldsymbol{f}_m^H\boldsymbol{\zeta}),
\end{align}
let $\frac{\partial L(\boldsymbol{\zeta},\lambda)}{\partial \boldsymbol{\zeta}}=0$ and $\frac{\partial L(\boldsymbol{\zeta},\lambda)}{\partial {\lambda}}=0$, i.e.,
\begin{align} \label{derivative}
\boldsymbol{\zeta}-{\nabla}_{\bar{R}_m}-\lambda\boldsymbol{f}_m=0,
\quad\boldsymbol{\zeta}^H\boldsymbol{f}_m+\boldsymbol{f}_m^H\boldsymbol{\zeta}=0.
\end{align}
Solving (\ref{derivative}), we can get
\begin{align} \label{Laresult}
\boldsymbol{\zeta}={\nabla}_{\bar{R}_m} - \mathfrak{R}[{({\nabla}_{\bar{R}_m})}^H\boldsymbol{f}_m]\boldsymbol{f}_m,
\quad\lambda=-\mathfrak{R}[{({\nabla}_{\bar{R}_m})}^H\boldsymbol{f}_m].
\end{align}
Thus, the closed form of ${\rm grad}_{\bar{R}_m}$ is ${\rm grad}_{\bar{R}_m}={\nabla}_{\bar{R}_m} - \mathfrak{R}[{({\nabla}_{\bar{R}_m})}^H\boldsymbol{f}_m]\boldsymbol{f}_m$, and the proof is completed.


\section{Proof of Lemma 3}
Since the columns of matrix $\mathbf{W}_{NP}\mathbf{U}$ are periodic, i.e., $[\mathbf{W}_{NP}\mathbf{U}]_{n,q}=[\mathbf{W}_{NP}\mathbf{U}]_{iN+n,q}, \forall n,q\in\{0,...,N-1\}, i\in\{0,...,P-1\}$, $\mathbf{W}_{NP}\mathbf{U}$ can be divided into $P$ subblocks by every $N$ rows, and each subblock is given by $\mathbf{W}_N$. Therefore, $\mathbf{Q}_m$ can be expressed as
\begin{align}
\mathbf{Q}_m=\frac{1}{N_0}\begin{bmatrix} \mathbf{W}_N & \cdots & \mathbf{W}_N \end{bmatrix}^T
\mathbf{C}_m^* \begin{bmatrix} \mathbf{W}_N^H & \cdots & \mathbf{W}_N^H \end{bmatrix},
\end{align}
which is a $P\times P$ block matrix, and each subblock is $\frac{1}{N_0}\mathbf{W}_N\mathbf{C}_m^*\mathbf{W}_N^H$. Thus, the problem is to show that $\frac{1}{N_0}\mathbf{W}_N\mathbf{C}_m^*\mathbf{W}_N^H$ is a diagonal matrix.

As the columns of $\mathbf{X}_m$ are orthogonal, it satisfies that $\mathbf{X}_m^H\mathbf{X}_m=\mathbf{\Lambda}_{\mathbf{X}_m}$, where $\mathbf{\Lambda}_{\mathbf{X}_m}\in\mathbb{R}^{N\times N}$ is a diagonal matrix of full rank. Thus, $(\mathbf{X}_m^H)^{-1}$ can be simplified as $(\mathbf{X}_m^H)^{-1}=(\mathbf{X}_m^{-1})^H=(\mathbf{\Lambda}_{\mathbf{X}_m}^{-1}\mathbf{X}_m^H)^H
=\mathbf{X}_m\mathbf{\Lambda}_{\mathbf{X}_m}^{-1}$, and its columns are also orthogonal. Since $\mathbf{W}_N^H$ is a unitary matrix, all columns of both $(\mathbf{X}_m^H)^{-1}$ and $\mathbf{W}_N^H$ can be regarded as two sets of orthogonal basis of the complex field $\mathbb{C}^N$. This indicates that columns of $(\mathbf{X}_m^H)^{-1}$ can be expressed as the linear combination of columns of $\mathbf{W}_N^H$, i.e., $(\mathbf{X}_m^H)^{-1}=\mathbf{W}_N^H\mathbf{K}_m$. Therefore, $\frac{1}{N_0}\mathbf{W}_N\mathbf{C}_m^*\mathbf{W}_N^H$ can be further simplified as
\begin{align}
\nonumber
&\frac{1}{N_0}\mathbf{W}_N\mathbf{C}_m^*\mathbf{W}_N^H
=\frac{1}{N_0}\mathbf{W}_N(\mathbf{X}_m^H)^{-1}\mathbf{S}_m^*(\mathbf{X}_m)^{-1}\mathbf{W}_N^H\\
=&\frac{1}{N_0}\mathbf{W}_N\mathbf{W}_N^H\mathbf{K}_m\mathbf{S}_m^*\mathbf{K}_m^H\mathbf{W}_N\mathbf{W}_N^H
=\frac{1}{N_0}\mathbf{K}_m\mathbf{S}_m^*\mathbf{K}_m^H,
\end{align}
which is a diagonal matrix, and the proof is completed.


\ifCLASSOPTIONcaptionsoff
  \newpage
\fi



%
%

\end{document}